\documentclass{lmcs} 
\usepackage[utf8]{inputenc}
\pdfoutput=1

\usepackage{lastpage}
\lmcsdoi{19}{2}{4}
\lmcsheading{}{\pageref{LastPage}}{}{}%
{Mar.~16,~2020}{Apr.~20,~2023}{}


\def\eg{{\em e.g.}}

\def\ie{{\em i.e.}}

\def\viz{{\em viz.}}

\usepackage{latexsym,amssymb}

\usepackage{url}
 \usepackage[all]{xy}
\usepackage{graphicx}
\usepackage{amsmath}
\usepackage{mathtools}
\usepackage{cmll}
\usepackage{stmaryrd}

\newcommand{\ve}[1]{*+[o][F]{#1}}
\newdir{C}{%
!/4.5pt/@{( }*:(1,-.2)@{}*:(1,+.2)@_{}}
\newdir{|>}{%
{}*!/-5pt/@{}*!/4.5pt/@{|}*:(1,-.2)@^{>}*:(1,+.2)@_{>}}
\newdir{pb}{:(1,-1)@^{|-}}
\def\pb#1{\save[]+<16 pt,0 pt>:a(#1)\ar@{pb{}}[]\restore}

\newcommand{\rstd}{\mathbin\upharpoonright}
\newcommand{\unamb}{{\it unamb}}
\newcommand{\col}{{\it col}}
\newcommand{\fame}{{\it fam}}
 
\newcommand{\ese}{\textrm{ese}}
\newcommand{\Ese}{\textrm{Ese}}
\newcommand{\er}{\textit{er}}
\newcommand{\edc}{\textrm{edc}}
\newcommand{\ef}{\textrm{ef}}
\newcommand{\ESE}{{\PES}_\equiv}
\newcommand{\PES}{\mathcal{E}}
\newcommand{\FAM}{\mathcal{F}\!am}
\newcommand{\SFAM}{\mathcal{S\!F}\!am}
\newcommand{\FAME}{\mathcal{F}\!am_\equiv}
\newcommand{\SFAME}{\mathcal{S\!F}\!am_\equiv}
\newcommand{\GES}{\mathcal{G}}
\newcommand{\EDC}{\mathcal{E\!D\!C}}
\newcommand{\Equiv}{\textrm{Equiv}}

\newcommand{\A}{{\mathcal A}}
\newcommand{\B}{{\mathcal B}}
\newcommand{\C}{{\mathcal C}}
\newcommand{\D}{{\mathcal D}}

\newcommand{\F}{{\mathcal F}}
\newcommand{\Con}{\textrm{Con}}
\newcommand{\id}{\textrm{id}}

\newcommand{\eeq}{\equiv}
\newcommand{\longcov}[1]{{\stackrel{#1}{\mathrel-\joinrel\relbar\joinrel\subset\,}}}
\renewcommand{\max}{\textit{top}}
\newcommand{\setdif}{\setminus}
\newcommand{\eps}{\, \epsilon}
\newcommand{\fsubseteq}{\subseteq_{\textrm{fin}}}

\newcommand{\conf}[1]{\:\!{\mathcal{C}}(#1)}
\newcommand{\iconf}[1]{\:\!{\mathcal{C}}^\infty(#1)}
\newcommand{\parrow}{\rightharpoonup}
\newcommand{\set}[2]{{\{  #1\  | \  #2 \} }}
\newcommand{\setof}[1]{{\{ #1 \} }}
\newcommand{\eqdef}{\mathrel{=_{\mathrm{def}}}}
\newcommand{\iso}{\cong}

\begin{document}

  \title{Causal Unfoldings and Disjunctive Causes}
  
  \author[M.~de Visme]{Marc de Visme}[a]	
  \address{Université Paris-Saclay, CNRS, ENS Paris-Saclay, Inria, LMF, France}	
  \email{marc.de-visme@inria.fr}  
  
  \author[G.~Winskel]{Glynn Winskel}[b]	
  \address{
  Edinburgh Research Centre, Central Software Institute, Huawei; 
University of Strathclyde, UK
}	
  \email{gw104@cl.cam.ac.uk}  
  
 
\begin{abstract}
In the simplest form of event structure, a prime event structure, an event is associated with a unique causal history, its prime cause. However, it is quite common for an event to have disjunctive causes in that it can be enabled by any one of multiple sets of causes. Sometimes the sets of causes may be mutually exclusive, inconsistent one with another,  and sometimes not, in which case they coexist consistently and constitute parallel causes of the event. The established model of general event structures can model parallel causes. On occasion however such a model abstracts too far away from the precise causal histories of events to be directly useful.  For example, sometimes one needs to associate probabilities with different, possibly coexisting, causal histories of a common event. Ideally, the causal histories of a general event structure would correspond to the configurations of its causal unfolding to a prime event structure; and the causal unfolding would arise as a right adjoint to the embedding of prime in general event structures. But there is no such adjunction.  However, a slight extension of  prime event structures remedies this defect and provides a causal unfolding as a universal construction.  Prime event structures are extended with an equivalence relation in order to dissociate the two roles, that of an event and its enabling; in effect, prime causes are labelled by a disjunctive event, an equivalence class of its prime causes.  With this  enrichment a suitable causal unfolding appears as a pseudo right adjoint. The adjunction relies critically on the central and subtle notion of extremal causal realisation as an embodiment of causal history. Finally, we explore subcategories which support parallel causes as well the key operations needed in developing 
probabilistic distributed strategies with parallel causes.  
 \end{abstract}

\maketitle

\section{Introduction}
Work on probabilistic distributed strategies based on event structures brought us face to face with a limitation in existing models of concurrent computation, and in particular with the theory of event structures as it had been developed.  In order to adequately express certain intuitively natural, optimal probabilistic strategies, it was necessary to simultaneously support: probability on event structures with opponent moves, itself rather subtle;  parallel causes, in which an event may be enabled in several distinct but compatible ways; and a hiding operation crucial in the composition of strategies.  The difficulties did not show up in the less refined development of nondeterministic strategies; there the simplest form of event structure, prime event structures, sufficed.  
The ``obvious'' remedy, to base strategies on more general event structures, which do support parallel causes, failed to support probability and hiding adequately.\footnote{In the context of games, {\em deterministic} general event structures do support parallel causes and hiding so strategies with parallel causes such as that of parallel-or, but of course not the underlying nondeterminism required for probabilistic strategies~\cite{fscd17,ecsym-notes}.}
The problems and a solution are documented in
the article~\cite{CSL17}.

That work uncovered a central construction, which we here call the {\em causal unfolding} of a model with parallel causes.  It  is based on the notion of  {\em extremal causal realisation} and attendant {\em prime  extremal realisation} which plays a role analogous to that of complete prime in distributive orders.  Both concepts deserve to be better known and are expanded on comprehensively with full proofs here.  As will  shortly be explained more fully, intuitively, a prime extremal realisation is a finite partial order expressing a minimal causal history for an event to occur, even in the presence of several parallel causes for the event. Extremal realisations provide us with a way to unfold a model supporting parallel causes (general event structures---Section~\ref{sec:GES}, or equivalence families---Section~\ref{sec:ESE})
  into a structure describing all its causal histories---its causal unfolding.  As is to be hoped, the unfolding will be a form of right adjoint giving the causal unfolding and extremal realisations a categorical significance.

To give an idea of  prime extremal realisations of events we give a short, necessarily informal, preview of two examples from the paper.  The simplest concerns a general event structure comprising three events $a$, $b$ and $d$ where $d$ can occur once $a$ or $b$ have occurred and where all events can occur together.  The events $a$ and $b$ constitute parallel causes of the event $d$.  We can picture the situation in the diagram:
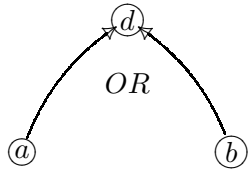
\[
\xymatrix@R=1em@C=2em{
 &\ve{d}&\\
&{OR}&\\
 \ve{a} \ar@{|>}@/^1pc/[uur]
& &\ve{b} \ar@{|>}@/_1pc/[uul]
}
\]
Here there are two minimal causal histories associated with the occurrence of the event $d$, \viz~$d$ after $a$, and $d$ after $b$ : 
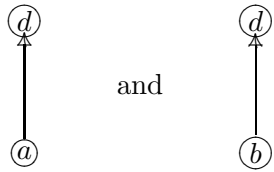
\[
\xymatrix@R=1em@C=2em{
 \ve{d}&&\ve{d}\\
& \hbox{ and }&\\
 \ve{a} \ar@{|>}[uu]
& &\ve{b} \ar@{|>}[uu]
}
\]
These will be the prime extremal realisations associated with the occurrence of $d$.  
 But this example is deceptively simple.  To add a level of difficulty, consider the general event structure 
 
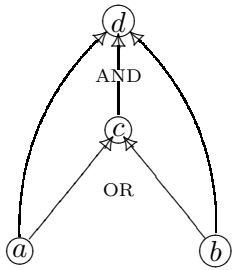
\[
\xymatrix@R=1em@C=2em{
{} &  \ve{d} & {}\\
{}&{}&{}\\
{}& \ve{c} \ar@{|>}[uu]|{\text{\tiny AND}} &{}\\
{}& {\text{\tiny OR}}&{}\\
\ve{a} \ar@{|>}[uur] \ar@{|>}@/^1pc/[uuuur] & {} & \ve{b} \ar@{|>}[uul] \ar@{|>}@/_1pc/[uuuul] \\ 
       }
\]
which portrays an event $d$ enabled through the occurrence of all of the events $a$, $b$ and $c$ but where $c$ is enabled by either $a$ or $b$.  This time the  two minimal causal histories associated with the occurrence of the event $d$, one after $c$ caused by $a$, and the other after $c$ caused by $b$, give rise to the two prime extremal realisations: 
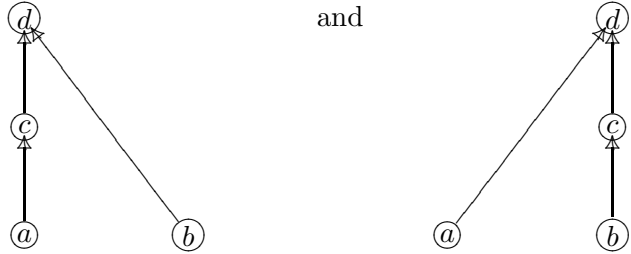
\[\xymatrix@R=1em@C=2em{
\ve{d}&&  {} &\\
{}&&{}&\\
\ve{c} \ar@{|>}[uu]&&  {}&\\
{}&&{}&\\
\ve{a} \ar@{|>}[uu] && \ve{b} \ar@{|>}[uuuull]&\\
        }
\xymatrix@R=1em@C=0em{
\\
\\
\hbox{and}\\
\\
\\
        }
\xymatrix@R=1em@C=2em{
&{}&& \ve{d} \\
&{}&&{}\\
&{}&& \ve{c} \ar@{|>}[uu]\\
&{}&&{}\\
& \ve{a} \ar@{|>}[uuuurr] && \ve{b} \ar@{|>}[uu]\\
        }\]
There are also more subtle `non-injective' prime extremal realisations in which the same event of a general event structure occurs in several different ways---see Example~\ref{ex:nonedc}, though these have been ruled out in our application to strategies with parallel causes~\cite{CSL17}. 
  
The new adjunction, with its right adjoint the causal unfolding,  supplies a missing link in the landscape of models for concurrency~\cite{WN}. The adjunction   
connects models with parallel causes, such as general event structures,  to those based on partial orders of events. It does this through the introduction of a simple, new model which is based on prime event structures extended with an equivalence relation on their sets of events. In its most general form the adjunction relates two new models, prime event structures with an equivalence relation to families of configurations with a similar equivalence.  We explore how the adjunction restricts and simplifies on subcategories, in particular, between {\em event structures with disjunctive causes (\edc's)} and {\em stable families with equivalence}.  Via the simplified adjunction we show that the category of \edc's has the constructions pullback, pseudo pullback and factorisation needed to develop distributed probabilistic strategies with parallel causes~\cite{CSL17}.  We point the reader to the figure of the Conclusion; it summarises the adjunctions relating the models we encounter and develop in the article. 

 More broadly, often  in  systems with parallel causes it   
is necessary to associate probabilities with causal histories, and the causal unfolding provides a suitable structure on which to do this systematically~\cite{CSL17}. Outside probability, there is a similar need for causal unfoldings, for example, when reversible computing encounters parallel causes~\cite{ioana,CKV}, and in extracting biochemical pathways, forms of causal history in biochemical systems where parallel causes are rife~\cite{kappa}.

\section{Event structures and their maps} 
We briefly review two well-established forms of event structure and explain the absence of an adjunction associated with the embedding of prime into general event structures. It is through such an adjunction one might otherwise have thought to find a causal unfolding of general event structures to prime event structures. The absence motivates a new model  
based on prime event structures with an equivalence relation. 
(We refer the reader to~\cite{thesis,evstrs} in particular for background and intuitions.)

\subsection{Prime event structures}\label{sec:PES}
A {\em prime event structure} comprises $(E, \leq, \Con)$,   consisting of a set $E$ of {\em events} which are partially ordered by $\leq$, the {\em causal dependency relation}, and a  non-empty {\em consistency relation} $\Con$ consisting of finite subsets of $E$.  The relation $e'\leq e$ expresses that event $e$ causally depends on the previous occurrence of event $e'$.  Write $[X]$ for the $\leq$-down-closure of a subset of events $X$. That a finite subset of events is consistent conveys that its events can occur together by some stage in the evolution of the process.  Together the relations satisfy several  axioms: 
\begin{center}
\begin{tabular}{l}
$[e]= \set{e'}{e'\leq e}$ is finite, for all $e\in E$,\\
$\setof{e}\in\Con$, for all $e\in E$,\\
$X\subseteq Y\in\Con \implies X\in \Con$, and \\
$X\in\Con \ \&\  e\leq e'\in X \implies X\cup\setof{e}\in\Con$.
\end{tabular}
\end{center}
A {\em configuration} is a, possibly infinite, set of events $x\subseteq E$ which is: 
{\em consistent,} $X\subseteq x  \hbox{ and }  X \hbox{ is finite}  \hbox{ implies } X\in\Con$\,; and {\em down-closed, } $[x]= x$.
It is part and parcel of prime event structures that an event   $e$ is associated with a unique causal history $[e]$.  

Prime event structures have a long history.  They first appeared in describing the patterns of event occurrences that occurred in the unfolding  of a (1-safe) Petri net~\cite{NPW}.  As their configurations, ordered by inclusion, form a Scott domain, prime event structures provided an early bridge between the semantic theories of Dana Scott and Carl Petri; one early result being that a {\em confusion-free} Petri net unfolded to a prime event structure with configurations taking the form of a {\em concrete} domain, as defined by Kahn and Plotkin.  Generally, the configurations of a countable prime event structure ordered by inclusion coincide with the dI-domains of Berry---distributed Scott domains which satisfy a finiteness axiom~\cite{evstrs}.   The domains of configuration of a prime event structure had been characterised earlier in~\cite{NPW} as {\em prime algebraic} domains, Scott domains with a subbasis of complete primes.\footnote{A {\em complete prime} in an order which supports least upper bounds $\bigsqcup X$ of compatible subsets $X$ is an element $p$ such that $p\sqsubseteq \bigsqcup  X$ implies $p\sqsubseteq x$ for some $x\in X$. In the configurations of a prime event structure the complete primes are exactly those configurations $[e]$ for an event $e$.}

 \subsection{General event structures}\label{sec:GES}

A  {\em general event structure}~\cite{thesis,evstrs} permits an event to be caused disjunctively in several ways, possibly coexisting in parallel, as parallel causes. A general event structure comprises $(E, \Con, \vdash)$  where $E$ is a set of events, the consistency relation $\Con$ is a non-empty collection of finite subsets of $E$, and the {\em enabling relation} ${\vdash}$ is a relation in $\Con\times E$ such that
\[\begin{matrix}
X\subseteq Y\in \Con \implies X\in\Con\,, \hbox{ and} \\
Y\in\Con\ \&\ Y\supseteq X \ \& \ X\vdash e \implies Y\vdash e\,.
\end{matrix}\]
A {\em configuration} is a subset $x$ of $E$ which is:  {\em consistent,}  $X\fsubseteq x \implies X\in\Con$; and {\em secured, }$\forall e\in x \exists  e_1, \cdots, e_n\in x.\ \ e_n=e\ \ \& \  \forall i\leq n.  \setof{e_1, \cdots, e_{i-1}}\vdash~e_i\,$. 
We write $\iconf E$ for the configurations of $E$ and $\conf E$ for its finite configurations.  (For illustrations of small general event structures see, for instance, Example~\ref{ex:docs} and $E_0$ of Example~\ref{Ex1}.
)

An event $e$ being enabled in a configuration has been expressed through the existence of a securing chain $e_1, \cdots, e_n$, with $e_n =e$, within the configuration. The chain represents a {\em complete enabling} of $e$ in the sense that every event in the chain is itself enabled by earlier members of the chain.  
Just as mathematical proofs are most usefully viewed not merely as
sequences, so 
later complete enablings  expressed more generally as partial orders
---``causal realisations''---will play a central role.  

A {\em map} $f:(E,\Con,\vdash)\to (E',\Con',\vdash')$ of general event structures  is a partial function $f:E\parrow E'$ such that 
 \[ \begin{matrix*}[l]
 \forall X\in\Con\,.\  fX \in \Con' \,, \\
 \forall X \in\Con, e_1, e_2\in X.\ f(e_1)=f(e_2)  \text{ (both defined)} \implies e_1=e_2\,, \hbox{ and }\\
  \forall  X \in\Con, e\in E.\ X\vdash e \ \&\ f(e)\text{  is defined }\implies fX\vdash' f(e)\,. \end{matrix*}
\]
 Maps compose as partial functions. Write $\GES$ for the category of general event structures. 

 W.r.t.~a family of sets $\F$, a subset $X$ of $\F$ is {\em compatible} (in $\F$), if there is $y\in \F$ such that $x\subseteq y$ for all $x\in X$; in particular $\setof{x,y}$ is compatible if there is $z \in \F$ such that $x,y \subseteq z$. Say a 
 subset is {\em finitely compatible} iff every finite subset is compatible.

We can now characterise those families of configurations arising from a  general event structure~\cite{evstrs}. 
A {\em family of configurations} comprises a non-empty family $\F$ of sets  such that if $X\subseteq\F$ is finitely compatible in $\F$ then $\bigcup X\in \F$; and if $e\in x\in \F$ there is a securing chain $e_1, \cdots, e_n=e$ in $x$ such that $\setof{e_1, \cdots, e_i} \in \F$ for all $i\leq n$.\footnote{The latter condition is equivalent to: (i) if $e\in x\in \F$ there is a finite $x_0\in \F$ s.t.~$e\in x_0\in \F$ and (ii) (coincident-freeness)  for distinct
$e, e'\in x$, there is $y\in \F$ with $y\subseteq x$ s.t.~$e\in y \iff e'\not\in y$.}
 Its {\em events} are elements of the underlying set $\bigcup \F$. 

A {\em map} between families of configurations from $\A$ to $\B$ is a partial function $f:\bigcup\A \parrow \bigcup\B$ between their events such that when $x\in\A$  then $f x\in\B$ and any event of $f x$  arises as the image of a unique event of $x$, \ie~the following {\em local injectivity} is fulfilled:
\[
e_1, e_2 \in x\ \&\ f(e_1) = f(e_2) \hbox{ (both defined) } \implies e_1 = e_2\,.
\]
 Maps compose as partial functions. Write $\FAM$ for the category of families of configurations. In Section~\ref{secn:stableefs}, we shall meet the subcategory $\SFAM$ of {\em stable families} of configurations, with objects $\A$ of $\FAM$ that satisfy
\[\forall x,y,z\in\A.\ \ x, y\subseteq z \     \Rightarrow\  x\cap y\in \A\,,\]
which  plays an important role in constructions on prime event structures~\cite{icalp82,evstrs}.

Characterisations of the orders obtained from the configurations of a general event structure can be found in~\cite{thesis}.\footnote{Complete irreducibles are the customary generalisation of complete primes to nondistributive orders such as those of configurations of general event structures ordered by inclusion~\cite{thesis}. A {\em complete irreducible} in an order which supports least upper bounds $\bigsqcup X$ of compatible subsets $X$ is an element $r$ such that $r=\bigsqcup  X$ implies $r= x$ for some $x\in X$. In the configurations of a general event structure the complete irreducibles are exactly those minimal configurations which contain an event $e$.
A forewarning: only in very special circumstances will prime  extremal realisations---the generalisation of complete prime of this paper---coincide with complete irreducibles---see Example~\ref{Ex1}.}

\subsection{A coreflection and non-coreflection}
There is a forgetful functor  $\GES\to \FAM$ taking a general event structure to its family of configurations. It has a left adjoint, which constructs a canonical general event structure from a family:  given $\A$, a family of configurations with underlying events $A$,   construct a general event structure $(A,\Con, \vdash)$ with $X\in \Con$ iff $X\fsubseteq y$, for some $y\in \A$; and with $X\vdash a$ iff $a\in A$, $X\in\Con$ and $a\in y  \subseteq X\cup\setof a$, for some $y\in \A$. The above yields a coreflection\footnote{A coreflection is an adjunction where the left adjoint is full and faithful, or equivalently the unit is iso.}
\[
\xymatrix{
 \FAM \ar@/_/[rr]_{}^{\top} &&   \GES  \ar@/_/[ll]_{} 
}\]
 of families of configurations in general event structures.  It cuts down to an equivalence between families of configurations and {\em replete} general event structures. A general event structure $(E,\Con, \vdash)$ is {\em replete} iff 
 \[
 \begin{matrix*}[l]
 \forall e\in E\,\exists X\in\Con.\ X\vdash e \  \,,\\
\forall X\in\Con\,\exists x\in\conf E.\  X\subseteq x \   \text{ and }  \\
 X\vdash e \implies \exists x\in \conf E.\ e\in x \ \&\ x\subseteq X\cup\setof e\,.
 \end{matrix*}
 \]
 
A {\em map} of prime event structures is a map of their families of configurations. Write $\PES$ for the category of prime event structures. (A map in $\PES$ need not preserve causal dependency; when it does and is total it is called {\em rigid}.) 
 
There is an obvious ``inclusion'' functor $\PES\to \FAM$ fully and faithfully embedding the category of prime event structures in the category of families of configurations and so in general event structures.
  We might expect the functor $\PES\to \FAM$ to be the left adjoint of a coreflection
\[\xymatrix{
\PES \ar@/_/[rr]_{}^{\top} &&  \FAM  \ar@{.>}@/_/[ll]_{\hbox{\bf ?}} \ar@/_/[rr]_{}^{\top} &&   \GES\,,  \ar@/_/[ll]_{} 
}\]
so yielding a composite right adjoint $\GES \to \PES$ which unfolds a  general event structure to a prime event structure~\cite{evstrs,WN}.
 However under reasonable assumptions this cannot exist, as the following example indicates.  

\begin{exa}\label{ex:docs}{\rm
Consider a general event structure comprising three events $a$, $b$ and $d$ with all subsets consistent and minimal enablings $\emptyset \vdash a, b$ and $\setof a\vdash d$ and $\setof b\vdash d$.  Imagine concurrent treatments $a$ and $b$ of two doctors which 
sadly lead to the death $d$ of the patient.
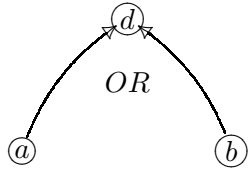
\[
\xymatrix@R=1em@C=2em{
 &\ve{d}&\\
&{OR}&\\
 \ve{a} \ar@{|>}@/^1pc/[uur]
& &\ve{b} \ar@{|>}@/_1pc/[uul]
}
\]
As its unfolding it is hard to avoid a prime event structure with events and causal dependency $a < d_a$ and $b < d_b$---the event $d_a$ representing ``death by $a$'' and the event $d_b$ ``death by $b$''---with the counit of the adjunction collapsing $d_a$ and $d_b$ to the common event $d$.  (If  
we are to apportion blame to the doctors we shall need the probabilities of $d_a$ and $d_b$ given $a$ and $b$~\cite{Pearl}.) 
In order for the counit to be a map we are forced to make $\setof{d_a, d_b}$ inconsistent. This is one issue: why should death by one doctor's treatment be in conflict with death by the other's---they could be jointly responsible?  But even more damningly the tentative counit fails the  universal property required of it!  Consider another prime event structure with three events comprising $a< d$ and $b<d$ (``death due to both doctors' treatments''). The obvious map to the family of configurations of the general event structure---the identity on events---fails to factor {\em uniquely} through the putative counit: $d$ can be sent to either $d_a$ or $d_b$; the event ``death by both doctors'' can be sent to either ``death by a'' or ``death by b.''   This raises the second issue: if we are to obtain the required universal property we have to regard these two maps as essentially the same. }
 \end{exa}

The two issues raised in the example suggest a common solution: to enrich prime event structures with   equivalence relations. This will allow a broader class of maps, settling the first issue, and introduce an equivalence on maps, settling the second.  The causal unfolding of the ``doctors example'' will be very simple and comprise the prime event structure  $a < d_a$ and $b < d_b$ with $d_a$ and $d_b$ equivalent events; with all events consistent.  In general the  construction of the unfolding is surprisingly involved; causal histories can be much more intricate than in the simple example.

\section{Events  with an equivalence, 
categories \texorpdfstring{$\ESE$}{E-equiv} and \texorpdfstring{$\FAME$}{Fam-equiv}}\label{sec:ESE}
  We build causal unfoldings in a new model, based on the obvious extension to events with an equivalence relation.
 A {\em (prime) event structure with 
equivalence} (an \ese)
is a structure
\[(P, \leq, \Con, \eeq)\]
where $(P, \leq, \Con)$ satisfies the axioms of a prime event structure and $\eeq$ is an equivalence relation on $P$.
The intention is that the events of $P$ represent {\em prime causes} while the $\eeq$-equivalence classes of $P$ represent {\em disjunctive events}: 
$p$ in $P$ is a prime cause of the event $\setof p_\eeq$.  Notice there may be several prime causes of the same event
and that these may be parallel causes in the sense that they are consistent with each other and causally independent.  
For the moment we do not impose any extra axioms.\footnote{Additional axioms will appear later motivated by results---see Section~\ref{sec:coreflnsinESE}. For example, although an \ese\, in general may have two equivalent events which are also
causally related, but this cannot hold in unfoldings based on extremal realisations.}

The extension by an equivalence relation on events is accompanied by an extension to families of configurations.
An {\em equivalence-family}  (ef) is a family of configurations $\A$ with an equivalence relation $\eeq_A$ on its underlying set $A \eqdef \bigcup\A$ (with no further axioms).  
Equivalence-families are the most general model we shall consider; they support parallel causes and, 
later, a causal unfolding.

Let $(\A, \eeq_A)$ and $(\B, \eeq_B)$ be ef's, with respective underlying sets $A$ and $B$.  A map $f: (\A, \eeq_A)\to(\B, \eeq_B)$ is a partial function $f:A\parrow B$ which  preserves $\eeq$, if $a_1\eeq_A a_2$ then either both $f(a_1)$ and $f(a_2)$ are undefined or both defined with $f(a_1)\eeq_B f(a_2)$, such that 
 \[x\in \A \implies f x\in \B\ \& \ 
\forall a_1, a_2 \in x. \   f(a_1)  \eeq_B f(a_2) \implies a_1\eeq_A a_2\,.
\]
Composition is composition of partial functions.  We regard two maps \[f_1,  f_2: (\A, \eeq_A)\to(\B, \eeq_B)\] as equivalent, and write $f_1\equiv f_2$,   iff they are equidefined and yield 
equivalent results, \ie\,
 if $ f_1(a)$ is defined then so is  $f_2(a)$ and $f_1(a) \eeq_B f_2(a)$, and  if $f_2(a$) is defined then so is  $f_1(a)$ and $f_1(a) \eeq_B f_2(a)$.  Composition respects $\equiv$. 
   This yields a 
  category of equivalence families $\FAME$; it is enriched in the category of sets with equivalence relations (also called setoids).\footnote{Appendix~\ref{app:equiv-enriched} provides background on categories enriched in equivalence relations.}

  Clearly from an ese~$(P, \eeq_P)$ we obtain an ef $(\iconf P, \eeq_P)$ and we take a map of ese's to be a map between their associated ef's.
  Write $\ESE$ for the category of \ese's; it too is enriched in the category of sets with equivalence relations. 
  When the equivalence relations $\eeq$ of  \ese's~are the identity we essentially have prime event structures and their maps. There is clearly a full-and-faithful embedding 
\[
\ESE \to \FAME\,,
\]
which preserves and reflects the equivalence on maps.  
  One virtue of \ese's is that they support a hiding operation, associated with a factorisation system~\cite{CSL17}. 
  
  We sometimes use an alternative description of their maps: 
  \begin{prop} \label{prop:charnmap}  A map of \ese's from $P$ to $Q$ is a
  partial function $f:P\parrow Q$  which preserves $\eeq$ such that 
\begin{enumerate}
\item[(i)]
for all $X\in \Con_P$ the direct image  
$fX \in \Con_Q$ and \\
$\forall p_1, p_2\in X. \ f(p_1)\eeq_Q  f(p_2) \implies p_1\eeq_P p_2$\,, and
\item[(ii)]
whenever $q\leq_Q f(p)$ there is  
$p'\leq_P p$ such that $f(p')=q$\,.
\end{enumerate}
\end{prop}

While an ese determines an ef, the converse, how to construct the causal unfolding of an ef to an ese, is much less clear.  To do so we follow up on the idea of Section~\ref{sec:GES} 
 of basing minimal complete enablings on partial orders.  A minimal complete enabling will correspond to a {\em prime extremal realisation}.  Realisations and extremal realisations are our next topic.  
 
 \section{Causal histories as extremal realisations}
  Extremal causal realisations formalise the notion of causal history in models with parallel causes, \viz~general event structures and the most general model of equivalence-families.  They will be the central tool in constructing the causal unfoldings of such models.    
  
  \subsection{Causal realisations}
 Let $\A$ be a family of configurations with underlying set $A$.  A {\em (causal) realisation} of $\A$ comprises a partial order $(R, \leq)$, its {\em carrier}, such that the set $\set{e'\in R}{e'\leq e}$ is finite for all events $e\in R$, together with a function $\rho:R\to A$ for which the image $\rho x \in \A$ when $x$ is a down-closed subset of $R$. We say a realisation is {\em injective}
 when it is injective as a function.

A map between realisations $(R,\leq), \rho $ and $(R',\leq'),\rho'$  is a partial surjective function $f:R\parrow R'$ which preserves down-closed subsets and satisfies $\rho(e) = \rho'(f(e))$ for all $ e\in R$ where $f(e)$ is defined. It  is convenient to write such a map as $\rho \succeq^f \rho'$. Occasionally we shall write $\rho \succeq \rho'$, or the converse $\rho' \preceq \rho$, to mean there is a map of realisations from $\rho$ to $\rho'$. 

A map of realisations $\rho \succeq^f \rho'$ factors into a ``projection'' followed by a total map 
\[\rho \succeq^{f_1}_1 \rho_0 \succeq^{f_2}_2 \rho'\,,\] 
where $\rho_0$ stands for the realisation $(R_0,\leq_0),\rho_0$ where $R_0 = \set{e\in R}{f(e) \hbox{ is defined}}\,$ is the domain of definition of $f$; $\leq_0$ is the restriction of $\leq$; $f_1$ is the inverse relation to the inclusion $R_0\subseteq R$;  and $f_2:R_0\to R'$ is the total part of function $f$. We are using $\succeq_1$ and $\succeq_2$ to signify the two kinds of maps. The $\succeq_1$-maps are reverse inclusions.  The $\succeq_2$-maps are exactly the total maps of realisations.  Total maps  $\rho \succeq_2^f \rho'$  are precisely those functions $f$ from the carrier of $\rho$ to the carrier of $\rho'$ which preserve down-closed subsets and satisfy $\rho = \rho' f$.

 \subsection{Extremal realisations}
Let $\A$ be a configuration family with underlying set $A$. We shall say a  realisation $\rho$ is {\em extremal} when $\rho \succeq_2^f \rho'$ implies $f$ is an isomorphism, for any realisation $\rho'$; it is called {\em prime extremal} when it in addition has a top element, \ie~its carrier contains an element which dominates all other elements in the carrier. Intuitively, an extremal realisation is a most economic
causal history associated with its image, a configuration of $\A$ ---it need not be unique; it is extremal in being a realisation with  minimal causal dependencies. 

Any realisation in $\A$ can be coarsened to an extremal realisation.

\begin{lem}\label{lem:existextr}  
For any realisation $\rho$ there is an extremal realisation $\rho'$ with $\rho\succeq_2^f \rho'$. 
\end{lem}
 \begin{proof}  The case when $(R,\leq),\rho$ is finite is straightforward. Assume there exists $\rho'$ such that $\rho\succeq_2^f \rho'$. The cardinality of $\set{(a,b)}{a \leq b}$ for $\rho$ is necessarily greater or equal to that of $\set{(a,b)}{a \leq' b}$ for $\rho'$; it is equal if and only if $\rho$ and $\rho'$ are isomorphic. So if we consider a sequence of non-isomorphic $\rho \succeq_2 \rho_1 \succeq_2 \dots$, this sequence is necessarily finite, of size at most the cardinality of $\set{(a,b)}{a \leq b}$. This ensures the existence of an extremal realisation in the finite case. The infinite case is more subtle, and relies on colimits and Zorn's lemma. 
 	
The category of realisations with total maps has colimits of total-order diagrams.  A diagram $d$ from a total order $(I,\leq)$ to realisations, comprises a collection of  total maps of realisations  $d_{i,j}:d(i)\to d(j)$ when $i\leq j$ s.t.~$d_{i,i}$ is always the identity map and 
if $i\leq j$ and $j\leq k$ then $d_{i,k} = d_{j,k}\circ d_{i,j}$.  We suppose each realisation $d(i)$ has carrier $(R_i, \leq_i)$ with $d(i):R_i\to A$.  
We construct the colimit realisation of the diagram as follows. 

The elements of the colimit realisation consist of
equivalence classes of elements of the disjoint union
$R\eqdef \biguplus_{i\in I} R_i$
under the equivalence
\[
(i,e_i) \sim (j, e_j) \iff \exists k\in I.\ i\leq k \ \&\ j\leq k\ \&\ d_{i,k}(e_i) = d_{j,k}(e_j) \,
\]
---we shall write $\setof{e_i}_\sim$ for the equivalence class of $(i,e_i)$.
Consequently we may define a function $\rho_R:R\to A$ by taking $\rho_R(\setof{e_i}_\sim)=  \rho_i(e_i)$.  
Because every $d_{i,j}$ is a surjective function,  every equivalence class in $R$ has a representative in $R_i$ for  every  $i\in I$.  Moreover, for any $e\in R$ there is $k\in I$ s.t.
\[\set{e'\in R}{e'\leq_R e} =  \set{\setof{e'_k}_\sim}{e'_k\leq_k e_k}\,,\]
where $e=\setof{e_k}_\sim$, so is finite---by the argument in the first paragraph of this proof.  It follows that $\rho_R$ is a realisation.  
The maps
 $f_i:\rho_i \succeq_2 \rho_R$, where $i\in I$, given by
$f_i(e_i) = \setof{e_i}_\sim$ form a colimiting cocone.  

Suppose $\rho$ is a realisation.  Consider all total-order diagrams $d$ from a total order $(I,\leq)$ to realisations starting from $\rho$ with $d_{i,j}$ not an isomorphism if $i<j$.  Within them one such diagram, $d'$ from $(I',\leq')$, extends another, $d$ from $(I,\leq)$, if the total order $(I,\leq)$ is an initial segment of the total order $(I',\leq')$ with $d$ a restriction of $d'$.
Amongst them, by  Zorn's lemma,  there is a maximal diagram w.r.t.~extension.  From the maximality of the diagram its colimit is necessarily  
extremal.  In more detail, by Zorn's lemma, construct a colimiting cocone $f_i:d(i) \succeq_2 \rho_R$, $i\in I$ ---using the same notation as above.  By maximality of the diagram some $f_k$ must be an isomorphism; otherwise we could extend the diagram by adding a top element to the total order and sending it to $\rho_R$. If $j$ should satisfy $k<j$ then  $f_j\circ d_{k,j} = f_k$ so $ f_k^{-1}\circ f_j\circ d_{k,j} =\id_{R_k}$. It would follow that $d_{k,j}$ is injective, as well as surjective, it being a total map of realisations, and  consequently that $d_{k,j}$ is an isomorphism---a contradiction.  Hence $k$ is the maximum element in $(I,\leq)$.  If the colimit were not extremal we could again adjoin a new top element above $k$ thus extending the diagram---a contradiction.   
 \end{proof}

For example, as a corollary, a
countable configuration of a family of configurations always has an injective extremal realisation.
By serialising the countable configuration,
$a_1\leq a_2\leq \cdots \leq a_n\leq \cdots$,
where $\setof{a_1, \cdots, a_n}\in \A$ for all $n$,
we obtain an injective realisation $\rho$.  By Lemma~\ref{lem:existextr} we can coarsen $\rho$ to an extremal realisation $\rho'$  with $\rho\succeq_2^f \rho'$.  As $\rho = \rho'f$ the surjective function $f$  is also injective, so a bijection, ensuring that the extremal realisation $\rho'$ is  injective.

The following rather technical lemma and corollary are crucial.

\begin{lem}\label{lem:forext}  Assume  $(R,\leq), \rho$,  $(R_0,\leq_0), \rho_0$ and $(R_1,\leq_1), \rho_1$ are realisations.  
\begin{enumerate}
	\item[(i)] Suppose $f=\rho \succeq_1^{f_1} \rho_0 \succeq_2^{f_2}\rho_1$.  Then there are maps so that $f=\rho \succeq_2^{g_2} \rho' \succeq_1^{g_1}\rho_1$:  
	\[
	\xymatrix{
		\rho \ar[d]_{f_1}\ar@{..>}[r]^{g_2}&\rho'\ar@{..>}[d]^{g_1} \\
		\rho_0 \ar[r]^{f_2}&\rho_1
	}
	\]
	\item[(ii)] Suppose $\rho \succeq_1^{f_1} \rho_0$ where $R_0$ is not a down-closed subset of $R$.  Then there are maps so that $f_1= \rho \succeq_2^{g_2} \rho' \succeq_1^{g_1}\rho_0$ with $g_2$ not an isomorphism:
	\[
	\xymatrix{
		\rho \ar[d]_{f_1}\ar@{..>}[r]^{g_2}&\rho'\ar@{..>}[dl]^{g_1} \\
		\rho_0 &
	}
	\] 
\end{enumerate}
\end{lem}
\begin{proof}
(i) Construct the realisation $(R',\leq'), \rho'$ as follows.  
Define
\[R' = (R\setdif R_0) \cup R_1\]
where w.l.o.g.~we assume the sets $R\setdif R_0$ and $R_1$ are disjoint.  Define  
$g_2:R\to R'$ to act as the identity on elements of $R\setdif R_0$ and as $f_2$ on elements of $R_0$.  

When $b\in R\setdif R_0$, define 
\[
a\leq' b
\  \text{ iff }\  
\exists a_0\in R.\ a_0 \leq b \ \&\ g_2(a_0) =a\,.
\] 
When $b\in  R_1$, define
\[
a\leq' b
 \ \text{ iff }\ 
 a\in R_1 \ \&\ a\leq_1 b
\,.
\]
To see $\leq'$ is a partial order observe that reflexivity and antisymmetry follow directly from the corresponding properties of $\leq$ and $\leq_1$.  Transitivity requires an argument by cases.  For example, in the most involved case, where  
\[
c\leq' a \hbox{  with } a\in R_1
 \hbox{ and }
a\leq' b \hbox{ with } b\in R\setdif R_0
\]
we obtain 
\[
c \leq_1 a \hbox{ and } a_0 \leq b
\] 
for some  
$a_0 \in R_0$ 
with  $f_2(a_0) = a$.
As $f_2$ is surjective and preserves down-closed subsets,
\[
c_0 \leq_0 a_0 \hbox{ and } a_0 \leq b
\] 
for some  
$c_0 \in R_0$ 
with $f_2(c_0) = c$.
Consequently, $c_0  \leq b$ with  $g_2(c_0) =c$, making $c\leq' b$, as required for transitivity.  

Define $\rho'$ to act as $\rho$ on elements of $R\setdif R_0$ and as $\rho_1$ on elements of $R_1$. Then $\rho = \rho' g_2$ directly.  We check $\rho'$ preserves down-closed subsets, so is a realisation.  Let $b\in R'$. We use, for instance, $[b]'$
to stand for the down-closure of $b$ according to $\leq'$.
If $b\in R_1$ then $\rho'[b]' = \rho_1 [b]_1\in \A$. If $b\in R\setdif R_0$ then $\rho'[b]' = \rho g_2 [b]$ is the image under $\rho$ of the down-closed subset $g_2 [b]$, so in $\A$.  
Because  $f_2$ preserves down-closed subsets  so does $g_2$.  
We already have  
$\rho = \rho' g_2$, making $g_2$ a map of realisations $\rho \succeq_2^{g_2} \rho'$. 
Define $g_1:R'\parrow R_1$ to be the reverse of the inclusion $R_1\subseteq R'$. Because  $\rho_1$ is the restriction of $\rho'$ to $R_1$,  
$g_1$ is a map of realisations $\rho' \succeq_1^{g_1}\rho_1$. By construction $f=g_1g_2$.

\noindent
(ii) This follows from the construction of $(R'\leq'),\rho'$ used in (i) but in the special case where $f_2$ is the identity map (with $R_0=R_1$).  Then $R'= R$ but $\leq'\neq \leq$ as there is $e\in R_0$ with $[e]_0 \subsetneq [e]$ ensuring that $[e]' = [e]_0 \neq [e]$.
 \end{proof}

\begin{cor}\label{cor:extsucceqext}
If $\rho$ is extremal and $\rho \succeq^f \rho'$, then $\rho'$ is extremal and there is $\rho_0$ s.t.~$f:\rho\succeq_1 \rho_0 \iso \rho'$.  Moreover, the carrier $R_0$ of $\rho_0$ is a down-closed subset of the carrier $R$ of $\rho$, with order the restriction of that on $R$.   
\end{cor}
\begin{proof} Directly from Lemma~\ref{lem:forext}.  Assume $\rho$ is extremal and $\rho \succeq^f \rho'$.  We can factor  $f$  into $\rho \succeq_1^{f_1} \rho_0 \succeq_2^{f_2}\rho'$.  From (i),  if   $\rho_0$ were not extremal nor would $\rho$ be---a contradiction; hence   $f_2$ is an isomorphism.  From (ii), the carrier $R_0$ of $\rho_0$ has to be a down-closed subset of the carrier $R$ of $\rho$, as otherwise we would contradict the extremality of $\rho$. 
 \end{proof}

It follows that if $\rho$ is extremal and $\rho\succeq^f \rho'$ then $\rho'$ is extremal and the  inverse  relation $g\eqdef f^{-1}$ is an injective function preserving and reflecting down-closed subsets, \ie~$g[r'] = [g(r')]$ for all $r'\in R'$.  In other words:

\begin{cor} \label{cor:extpreceqextrigid} If $\rho$ is extremal and $\rho \succeq^f \rho'$, then $\rho'$ is extremal and the inverse $g\eqdef f^{-1}$ is a rigid embedding  from the carrier of $\rho'$ to the carrier of $\rho$ such that $\rho' = \rho g$.\footnote{Rigid embeddings first appeared in work of Gilles Kahn and Gordon Plotkin; they are the embeddings appropriate to stable domain theory---see \eg~Section~1.6 of~\cite{evstrs}.}
\end{cor}

\begin{lem}\label{lem:extperfect}
Let $(R,\leq), \rho$ be an extremal realisation.  Then 
\begin{enumerate}
\item[(i)] if $r'\leq r$ and $\rho(r) =\rho(r')$ then $r=r'$; 
\item[(ii)] if $[r) = [r')$ and  $\rho(r) =\rho(r')$ then $r=r'$. Here $[r)\eqdef [r]\setdif\setof r$.
\end{enumerate}
\end{lem}
\begin{proof} (i) Suppose  $r'\leq r$ and $\rho(r) =\rho(r')$.  By Corollary~\ref{cor:extpreceqextrigid}, we may project to $[r]$ to obtain an extremal realisation $\rho_0:[r]\to A$.
Suppose $r$ and $r'$ were unequal.  
We can define a realisation as the restriction of $\rho_0$ to $[r)$.  The function from $[r]$ to $[r)$ taking $r$ to $r'$ and otherwise acting as the identity function is a map of realisations from the realisation $\rho_0$ and clearly not an isomorphism, showing $\rho_0$ to be non-extremal---a contradiction.  Hence $r=r'$, as required.  

\noindent
(ii) Suppose $[r) = [r')$ and  $\rho(r) =\rho(r')$.  Projecting to $[\setof{r, r'}]$ we obtain an extremal realisation. If $r$ and $r'$ were unequal  there would be a non-isomorphism map to the realisation obtained by projecting to $[r]$, \viz~the map from $[\setof{r, r'}]$ to $[r]$ sending $r'$ to $r$ and fixing all other elements.  
 \end{proof}

In fact, by modifying condition (i) in the lemma above a little we can obtain a characterisation of extremal realisations---though not strictly necessary for the rest of  of the paper:

\begin{lem}\label{lem:charnofextremals}
Let $(R,\leq), \rho$ be a realisation.  Then $\rho$ is extremal iff
\begin{enumerate}
\item[(i)] if $X\subseteq [r)$, with $X$ down-closed and $r\in R$, and $\rho(X\cup\setof r) \in \A$ then $X=[r)$; and 
\item[(ii)] if $[r) = [r')$ and  $\rho(r) =\rho(r')$ then $r=r'$.
\end{enumerate}
\end{lem}
\begin{proof}  
\noindent
{\it ``Only if''}:  Assume $\rho$ is extremal.  We have already established (ii) in
 Lemma~\ref{lem:extperfect}.  To show (i), suppose  $X$ is down-closed  and $X\subseteq [r)$ in $R$ with $\rho(X\cup\setof r) \in \A$.  By Corollary~\ref{cor:extpreceqextrigid}, we may project to $[r]$ to obtain an extremal realisation $\rho_0:[r]\to A$.
Modify the restricted order $[r]$ to one in which $r'\leq r$ iff $r'\in X\cup \setof r$, and is otherwise unchanged.  The same underlying function $\rho_0$ remains a realisation, call it $\rho_0'$, on the modified order.  The identity function gives us a map $f:\rho_0 \succeq_2 \rho_0'$ which is an isomorphism between realisations iff $X= [r)$.

\noindent
{\it ``If''}:
Assume (i) and (ii).  Suppose  $f: \rho\succeq_2 \rho'$, where $R', \rho'$ is a realisation.  We show $f$ is injective and order-preserving.  As $f$ is presumed to be surjective and to preserve down-closed subsets we can then conclude it is an isomorphism.  

To see $f$ is injective suppose the contrary that $f(r_1) = f(r_2)$ for $r_1\neq r_2$.  W.l.o.g.~we may suppose $r_1$ and $r_2$ are minimal in the sense that
\[
r_1'\neq r_2'\ \&\ r_1'\leq r_1 \ \&\  r_2' \leq r_2  \ \&\   f(r_1') = f(r_2') \implies r_1'= r_1 \ \&\ r_2'= r_2\,.
\]
Define $r'\eqdef  f(r_1) = f(r_2)$. Then
\[
[r'] \subseteq f[r_1] \ \&\ [r'] \subseteq f[r_2]\,.
\]
Furthermore, 
as only $r'$ can be the image of $r_1$ and $r_2$ under the function $f$, 
 \[
[r') \subseteq f[r_1) \ \&\ [r') \subseteq f[r_2)\,.
\]
It follows that 
\[
 [r')\subseteq f[r_1) \cap  f[r_2) = f ([r_1)\cap [r_2))
\]
where the equality is a consequence of the minimality of $r_1, r_2$.  
Taking $X\eqdef [r_1)\cap [r_2)$ we have $(fX)\cup\setof{r'}$ is down-closed in $R'$.
Therefore
\[
\rho(X\cup \setof{r_1}) = \rho'f (X\cup \setof{r_1}) = \rho' (fX \cup\setof{r'}) \in \A\,.
\]
By condition (i), $X=[r_1)$.  Similarly, $X=[r_2)$, so $[r_1) = [r_2)$.  Obviously $\rho(r_1) = \rho'f (r_1) =  \rho'f (r_2)  = \rho(r_2)$, so we obtain $r_1= r_2$ by (ii) ---a contradiction, so $f$ is injective.

We now check that $f$ preserves the order.  Let $r\in R$. Define
\[
X\eqdef [\set{r_1}{r_1\leq r\ \&\ f(r_1)< f(r)}]\,,
\]
where the square brackets signify down-closure in $R$.  
Then $X$ is down-closed in $R$ by definition and $X\subseteq [r)$.  We have
$[f(r)] \subseteq f[r]$
whence
\[
f X = f [r] \cap [f(r)) =  [f(r))\,.
\]
Therefore
$f X \cup\setof{f(r)}$ is down-closed in $R'$, so
\[
\rho(X\cup\setof r) = \rho'f(X\cup\setof r) =  \rho'(fX\cup\setof{ f(r)}) \in \A\,.
\]
Hence $X= [r)$, by (i).  It follows that
\[
r_1 < r 
\implies r_1\in X
\implies f(r_1) < f(r) \hbox{ in } R'\,.
\]
This  shows that $f$ preserves the order on $R$.
 \end{proof}
 
\begin{lem}
There is at most one map between extremal realisations.
\end{lem}
\begin{proof}
Let $(R,\leq), \rho$  and $(R',\leq'), \rho'$ be extremal realisations.  Let $f, f':\rho \to \rho'$ be  maps with converse relations $g$ and $g'$ respectively.  We show the two functions $g$ and $g'$ are equal, and hence so are their converses $f$ and $f'$.  Suppose otherwise that $g\neq g'$. Then there is an $\leq$-minimal $r'\in R'$ for which 
$g(r')\neq g'(r')$ and $g[r') = g'[r')$.  Hence $[g(r')) = [g'(r'))$ and $\rho(g(r')) =\rho'(r') = \rho(g'(r'))$. As $\rho$ is extremal, by Lemma~\ref{lem:extperfect}(ii) we obtain $g(r') = g'(r')$---a contradiction.
\end{proof}

Hence, by the lemma above, extremal realisations of $A$ under $\preceq$ form a preorder.  We define the {\em order of extremal realisations} to have elements isomorphism classes of extremal realisations ordered according to the existence of a map between representatives of isomorphism classes. Alternatively, and equivalently, we could take a choice of representative from each isomorphism class and order these according to whether there is a map from one to the other.  
Recall  a prime extremal realisation is an extremal realisation with a  top element, \ie~when its carrier contains an element which dominates all other elements in the carrier.  The following proposition 
is a direct corollary of Proposition~\ref{prop:domianofrealzns} in the next section.

\begin{prop}  The order of extremal realisations of a 
family of configurations 
$\A$ forms a prime-algebraic domain~\cite{NPW} with complete primes
the prime extremal realisations.
\end{prop}

The proofs of the following observations are straightforward consequences of the definitions.  They
  emphasise that 
prime extremal realisations  are  a  
generalisation of complete primes.

\begin{prop}\label{prop:extremalsofprevstrsareconfigigs}  
Let $(A, \leq_A, \Con_A)$ be a prime event structure. 
The function which takes  an extremal realisation $(R,\leq_R), \rho$ to the configuration $\rho R$ is an order isomorphism from the order of extremal realisations of $\iconf A$ to the order of configurations $(\iconf A, \subseteq)$;  prime extremal realisations   correspond to complete primes of $\iconf A$. 
For an individual extremal realisation $(R,\leq_R), \rho$, the function $r\mapsto \rho(r)$ is an order isomorphism between $(R,\leq_R)$ and  $(\rho R,\, \leq_A\!\!\rstd \rho R)$, the order of the event structure restricted to the configuration $\rho R$.
\end{prop}  
 
A configuration $x\in\F$, of a family of configurations $\F$, is {\em irreducible} iff there is a necessarily unique $e\in x$ such that $\forall y\in \F,\ e\in y \subseteq x$ implies $y=x$. Irreducibles coincide with complete (join) irreducibles w.r.t.~the order of inclusion. It is tempting to think of irreducibles as representing minimal complete enablings.  But, as sets, irreducibles both (1) lack sufficient structure:  in the  formulation we are led to, of minimal complete enablings as prime extremal realisations, several prime realisations can have the same irreducible as their underlying set; and (2) are not general enough: there are prime realisations whose underlying set is not an irreducible.  
We conclude with examples illustrating the nature of extremal realisations; it is convenient to describe families of configurations by general event structures. 
\begin{exa}\label{Ex1}{\rm
This example shows that prime extremal realisations  do not correspond to irreducible configurations. First, we show a general event structure $E_0$ (all subsets consistent) with irreducible configuration $\setof{a,b,c,d}$ and  two (injective) prime extremals $E_1$ and $E_2$ with tops $d_1$ and $d_2$ which both have the same irreducible configuration $\setof{a,b,c,d}$ as their image. The lettering indicates the functions associated with the realisations, \eg~events $d_1$ and $d_2$ in the partial orders map to $d$ in the general event structure.
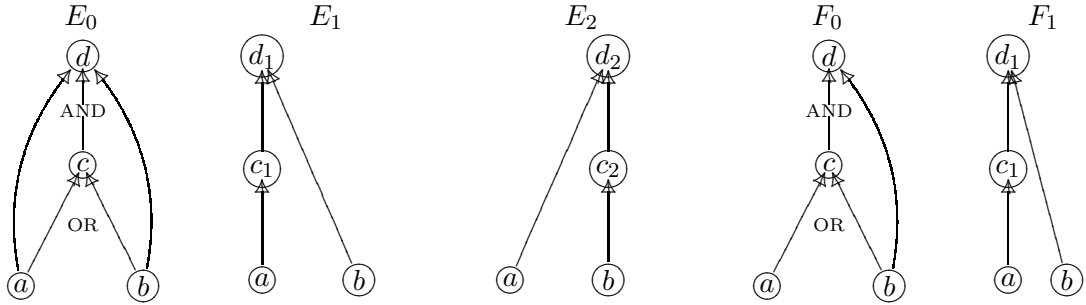
\begin{center}
\begin{tabular}{ccccccccc}
$E_0$ && $E_1$ && $E_2$ && $F_0$ && $F_1$ \\
$\xymatrix@R=1em@C=8pt{
{} & \ve{d} & {}\\
{}&{}&{}\\
{}& \ve{c} \ar@{|>}[uu]|{\text{\tiny AND}} &{}\\
{}& {\text{\tiny OR}}&{}\\
\ve{a} \ar@{|>}[uur] \ar@{|>}@/^1pc/[uuuur] & {} & \ve{b} \ar@{|>}[uul] \ar@{|>}@/_1pc/[uuuul] \\ 
       }$&$\quad$&
$\xymatrix@R=1em@C=8pt{
\ve{d_1}&&  {} \\
{}&&{}\\
\ve{c_1} \ar@{|>}[uu]&&  {}\\
{}&&{}\\
\ve{a} \ar@{|>}[uu] && \ve{b} \ar@{|>}[uuuull]\\
        }$&$\quad$&
$\xymatrix@R=1em@C=8pt{
{}&& \ve{d_2} \\
{}&&{}\\
{}&& \ve{c_2} \ar@{|>}[uu]\\
{}&&{}\\
 \ve{a} \ar@{|>}[uuuurr] && \ve{b} \ar@{|>}[uu]\\
        }$&$\quad$&
$\xymatrix@R=1em@C=8pt{
{} & \ve{d} & {}\\
{}&{}&{}\\
{}& \ve{c} \ar@{|>}[uu]|{\text{\tiny AND}} &{}\\
{}&{\text{\tiny OR}}&{}\\
\ve{a} \ar@{|>}[uur]
& {} & \ve{b} \ar@{|>}[uul] \ar@{|>}@/_1pc/[uuuul] \\ 
        }$&$\quad$&
$\xymatrix@R=1em@C=8pt{
\ve{d_1}&&  {} \\
{}&&{}\\
\ve{c_1} \ar@{|>}[uu]&&  {}\\
{}&&{}\\
\ve{a} \ar@{|>}[uu] && \ve{b} \ar@{|>}[uuuull]\\
				}$\\
\end{tabular}
\end{center}
 On the other hand there are prime extremal realisations   of which the image is   not   an irreducible configuration.  Consider the general event structure $F_0$. The prime extremal  $F_1$ describes a situation where $d$ is enabled by $b$ and $c$, and $c$ is enabled by $a$. It has image the configuration $\setof{a,b,c,d}$ which is not irreducible, being the union of the two incomparable configurations $\setof a$ and $\setof{b,c,d}$. 
} \end{exa}

\begin{exa}{\rm  \label{ex:nonedc}
It is possible to have extremal realisations in which an event depends on an event of the family having been enabled in two distinct ways, as in the following prime extremal realisation, on the left; it is clearly not injective.  
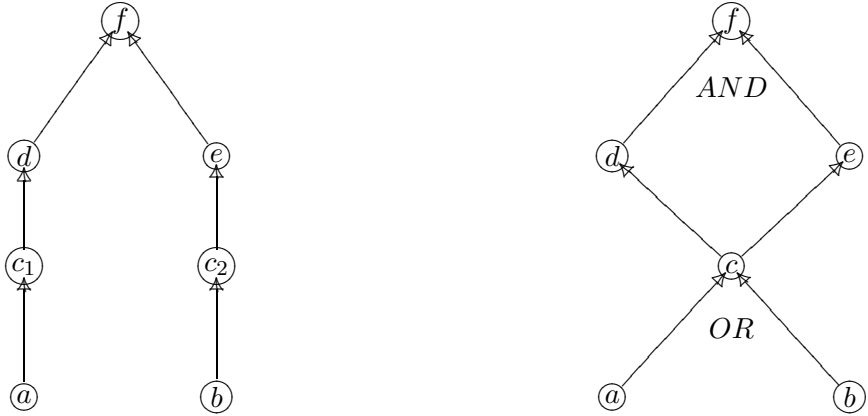
\[\xymatrix@R=1em@C=2em{
{} & \ve{f} & {} &
 &&&& {} & \ve{f} & {} \\
{} & {} & {} &
&&&& {} & AND & {} \\
\ve{d} \ar@{|>}[uur] & {} & \ve{e} \ar@{|>}[uul] &
&&&&\ve{d} \ar@{|>}[uur] & {} & \ve{e} \ar@{|>}[uul] \\
{} & {} & {} &
&&&& {} & {} & {} \\
\ve{c_1} \ar@{|>}[uu]  & {} & \ve{c_2} \ar@{|>}[uu] &
&&&& {} & \ve{c} \ar@{|>}[uur] \ar@{|>}[uul] & {} \\
{} & {} & {} &
&&&& {} & OR & {} \\
\ve{a} \ar@{|>}[uu] & {} & \ve{b} \ar@{|>}[uu] &
&&&&\ve{a} \ar@{|>}[uur] & {} & \ve{b} \ar@{|>}[uul]
}\]
The extremal describes the event $f$ being enabled by $d$ and $e$ where they are in turn enabled by different ways of enabling $c$.  We assume all subsets consistent.
} \end{exa}
\section{The causal unfolding: an adjunction from \texorpdfstring{$\ESE$}{E-equiv} to \texorpdfstring{$\FAME$}{Fam-equiv}
}
Furnished with the concept of extremal realisation, 
we can now exhibit an adjunction (precisely, a very simple case of biadjunction or pseudo adjunction) from $\ESE$, the category of \ese's, to 
 $\FAME$, the category of equivalence families. 
The left adjoint $I:\ESE\to \FAME$ is the full and faithful functor which takes an ese to its family of configurations with the original equivalence.  

The right adjoint, the {\em causal unfolding}, $\er:\FAME \to \ESE$ is defined on objects as follows.  
Let $\A$ be an equivalence family with underlying set $A$.  
Define $\er(\A) = (P,\Con_P, \leq_P, \equiv_P)$ where
\begin{itemize}
\item
$P$ consists of a choice from within each isomorphism class of the prime extremals $p$ of $\A$ ---we write $\max(p)$ for the image of the top element in $A$;
\item
Causal dependency $\leq_P$ is $\preceq$ on $P$;
\item
$X\in \Con_P$ iff $X\fsubseteq P$ and $\max\, [X]_P \in \A$ ---the set $[X]_P$ is the $\leq_P$-downwards closure of $X$, so equal to $\set{p'\in P}{\exists p\in X.\ p'\preceq p}$, and  $\max\, [X]_P$ is its image under $\max$;
\item
$p_1\equiv_P p_2$ iff $p_1, p_2\in P$ and $\max(p_1) \equiv_A \max(p_2)$.
\end{itemize}

\begin{prop} \label{prop:domianofrealzns}
The configurations of $P$ defined above, ordered by inclusion, are order-isomorphic to the order of extremal realisations: an extremal realisation $\rho$ corresponds, up to isomorphism, to the configuration $\set{p\in P}{p\preceq \rho}$ of $P$;  conversely, a configuration $x$ of $P$ corresponds to an extremal realisation
$\max:x  \to A$ with carrier $(x,\preceq)$, the restriction of the order of $P$ to $x$.
\end{prop}
\begin{proof} It will be helpful to recall, from Corollary~\ref{cor:extpreceqextrigid}, that if $\rho \succeq^f \rho'$ between extremal realisations, then the inverse relation $f^{-1}$ is a rigid embedding of (the carrier of) $\rho'$ in (the carrier of) $\rho$; so $\rho'\preceq \rho$ stands for a rigid embedding.
Suppose $x\in\iconf P$.  
Then $x$ determines an extremal realisation
\[
\theta(x)\eqdef \max:(x, \preceq) \to A\,.
\]
The function $\theta(x)$ is a realisation because each $p$ in $x$ is, and extremal because, if not, one of the $p$ in $x$ would fail to be extremal, a contradiction. Clearly $\rho'\preceq \rho$ implies $\theta(\rho')\subseteq \theta(\rho)$. 
Conversely, it is easily checked that any extremal realisation $\rho:(R,\leq) \to A$ defines  a configuration $\set{p\in P}{p\preceq \rho}$.  If $x\subseteq y$ in $\iconf P$ then $\phi(x)\preceq \phi(y)$.  It can be checked that $\theta$ and $\phi$ are mutual inverses, \ie~$\phi\theta(x) = x$ and $\theta\phi(\rho) \iso \rho$ for all configurations $x$ of $P$ and extremal realisations $\rho$.
 \end{proof}

From the above proposition we see that the events of $\er(\A)$ correspond to the order-theoretic completely-prime extremal realisations~\cite{NPW}. This justifies our  use of the term `prime extremal' 
for  extremal with top element.

The component of the counit of the adjunction
  $\eps_A:I(\er(\A)) \to \A$  is given by the function 
\[
\eps_A(p) =   \max(p)\,.
\]
It is a routine check to see that $\eps_A$ preserves $\eeq$ and that any configuration $x$ of $P$ images under $\max$ to a configuration in $\A$,   moreover in a way that reflects $\equiv$.

\begin{thm}
Let $\A\in \FAME$.  For all $f: I(Q)\to \A$ in $\FAME$, there is a map $h:Q\to \er(\A)$ in $\ESE$ such that 
$f = \eps_A\circ I(h)$,
\ie~so the diagram
\[
\xymatrix{
\A & \ar[l]_{\eps_A} I(\er(\A))\\
& \ar[ul]^f I(Q) \ar@{..>}[u]_{I(h)}
}\]
commutes.  Moreover, if $h':Q\to \er(\A)$ is a map in $\ESE$ s.t.~$f \equiv \eps_A\circ I(h')$, \ie~the diagram above commutes up to $\equiv$,  then $h'\equiv h$.
\end{thm}
\begin{proof}

Let $Q =(Q, \Con_Q, \leq_Q, \equiv_Q)$ be an \ese~and $f:I(Q) \to \A$ a map in $\FAME$.  We shall define  a  map $h:Q\to \er(\A)$ s.t.~$f= \eps_A h$. (As here, in the proof we shall elide the composition symbol $\circ$, and $I$ on maps which it leaves unchanged.)

We define the map $h:Q\to \er(\A)$ by induction on the depth of $Q$. The depth of an event in an event structure is the length of a longest $\leq$-chain up to it---so an initial event has depth 1. We take the depth of an event structure to be the maximum depth of its events.
(Because of our reliance on Lemma~\ref{lem:existextr}, we use the axiom of choice implicitly.)

Assume inductively that $h^{(n)}$ defines a map from $Q^{(n)}$ to $\er(\A)$ where $Q^{(n)}$ is the restriction of $Q$ to depth below or equal to $n$ such that $f^{(n)}$ the restriction of $f$ to $Q^{(n)}$ satisfies $f^{(n)}= \eps_A h^{(n)}$.  (In particular, $Q^{(0)}$ is the empty \ese~and $h^{(0)}$ the empty function.)  Then, by Proposition~\ref{prop:domianofrealzns}, any configuration $x$ of $Q^{(n)}$ 
determines
an extremal realisation $\rho_x: h^{(n)} x \to A$ with carrier $(h^{(n)}  x, \preceq)$.  

Suppose $q\in Q$ has depth $n+1$.  If $f(q)$ is undefined take $h^{(n+1)}(q)$ to be undefined.  Otherwise, note there is an extremal realisation $\rho_{[q)}$ with carrier $(h[q), \preceq)$.   Extend $\rho_{[q)}$ to a realisation $\rho_{[q)}^\top$  with carrier that of  $\rho_{[q)}$ with a new top element $\top$ adjoined, and make 
$\rho_{[q)}^\top$ extend the function $\rho_{[q)}$ by taking $\top$ to $f(q)$.  By Lemma~\ref{lem:existextr}, there is 
an extremal realisation $\rho$ such that $\rho_{[q)}^\top\succeq_2 \rho$.  Because $\rho_{[q)}$ is extremal,
$\rho\succeq_1 \rho_{[q)}$,  so $\rho$ only extends the order of $\rho_{[q)}$ with extra dependencies of $\top$. (For notational simplicity we  identify the carrier of $\rho$ with the set $h[q)\cup\setof\top$.)
Project $\rho$ to the extremal with top  $\top$.  Define this to be the value of $h^{(n+1)}(q)$.  In this way, we extend $h^{(n)}$ to a partial function $h^{(n+1)}: Q^{(n+1)} \to \er(\A)$ such that  $f^{(n+1)}= \eps_A h^{(n+1)}$. To see that $h^{(n+1)}$ is a map we can use Proposition~\ref{prop:charnmap}.  By construction  $h^{(n+1)}$ satisfies  property (ii) of Proposition~\ref{prop:charnmap} and the other properties are inherited fairly directly from $f$ via the definition of $\er(\A)$.

Defining $h=\bigcup_{n\in\omega} h^{(n)}$ we obtain a map 
$h : Q  \to \er(\A)$ such that  $f = \eps_A h$.

Suppose $h':Q\to \er(\A)$ is a map
s.t.~$f \equiv \eps_A  h'$.  Then, for any $q\in Q$,
\[
\max(h'(q)) = \eps_A h'(q) \eeq_A f(q) = \eps_A h (q) = \max(h(q))\,,
\]
so $h'(q) \eeq_P h(q)$ in $\er(\A)$.  Thus $h'\eeq h$.  
 \end{proof}

The theorem does not quite exhibit a traditional adjunction, because the usual cofreeness condition specifying an adjunction is weakened to only having uniqueness up to $\equiv$.  However the condition it describes does specify an exceedingly simple case of a  pseudo adjunction (or biadjunction) between 2-categories---a set together with an equivalence relation (a {\em setoid}) is a very simple example of a category.  As a consequence, whereas the usual cofreeness condition allows us to extend the right adjoint to arrows, so obtaining a functor, in this case following that same line will only yield a pseudo functor $\er$ as right adjoint: thus extended, $\er$ will only preserve composition and identities up to $\equiv$.  

The map
$(P,\eeq) \to \er(\iconf P, \eeq)$
which takes 
$p\in P$
to
the realisation with carrier  
$([p], \leq)$, the restriction of the causal dependency of $P$, with
the inclusion function 
$[p] \hookrightarrow P$
is an isomorphism; recall from Proposition~\ref{prop:extremalsofprevstrsareconfigigs} that the  configurations of a prime event structure correspond to its extremal realisations.  Such maps furnish the components of the unit of the pseudo adjunction:
\[
\xymatrix{
\ESE \ar@/_/[rr]_{I}^{\top} &&  \FAME \ar@/_/[ll]_{\er} 
}
\]

\begin{exa}\label{Ex3}{\rm
On the right we show a general event structure (all subsets consistent) and on its left its causal unfolding to an \ese\ under $\er$\,; the unfolding's events are the prime extremals.\footnote{See~\cite{MarcRep} for further examples of the causal unfolding including an inductive characterisation in~5.2.2.}
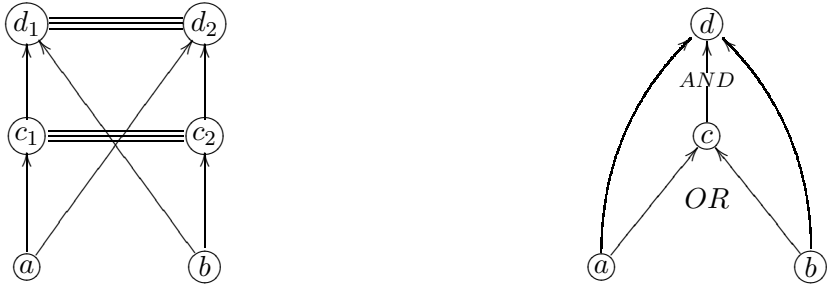
\[ \xymatrix@R=1em@C=2em{
 \ve{d_1} \ar@3{-}[rr] & {} & \ve{d_2}
{} &  {} & {} & {} & {} &
{} & \ve{d} & 
                \\
{} & {} & {} &
        {} & {} & {} & {} &
                {} & {} & {} \\            
                \ve{c_1} \ar[uu]^{
                } \ar@3{-}[rr] & {} & \ve{c_2} \ar[uu]_{
                }
& {}
& {} & {} & {} &
 {} & \ve{c} \ar[uu]|{AND
} & {} 
                \\
{} & & {} &
        {} & {} & {} & {} &
                {} & {OR}  & {} \\
 \ve{a} \ar[uu] \ar[uuuurr] & {} & \ve{b} \ar[uu] \ar[uuuull]
& {} & {} & {} & {} &
 \ve{a} \ar[uur] \ar@/^1pc/[uuuur] & {} & \ve{b} \ar[uul] \ar@/_1pc/[uuuul]        
        }\]  
} \end{exa}

\section{
Unfolding general event structures
}\label{sec:adjesetoges}

Recall $\GES$ is the category of general event structures.
We obtain a pseudo adjunction from $\ESE$ to $\GES$ via
an adjunction from $\FAME$ to $\GES$.
The right adjoint  $\fame:\GES\to\FAME$ is most simply described.  Given $(E, \Con, \vdash)$ in $\GES$ it returns the  equivalence family $(\iconf E, =)$  in $\FAME$ comprising the configurations together with the identity equivalence between events that appear within some configuration; the partial functions between events that are maps in $\GES$ are automatically maps in $\FAME$---the action of $\fame$ on maps.

For the effect of the left adjoint $\col:\FAME\to \GES$ on objects, define the {\em collapse} 
\[\col(\A) \eqdef (E,\Con, \vdash)\]  
where
\begin{itemize}
\item
$E = A_\equiv$, the equivalence classes of events in $A \eqdef \bigcup \A$\,;
\item
$X\in \Con$ iff $X\fsubseteq y_\equiv\eqdef\set{\setof a_\equiv}{a\in y}$, for some $y\in \A$\,; and
\item
$X\vdash e$ iff $e\in E$, $X\in\Con$ and $e\in y_\equiv \subseteq X\cup\setof e$, for some $y\in \A$.
\end{itemize}
It  follows that $y_\equiv$ is a configuration of $\col(\A)$ whenever $y\in\A$.  From this it is easy to see that $\col(\A)$ is a replete general event structure.

Let $(\A, \equiv)\in\FAME$. Assume that $\A$ has underlying set $A$.   The unit of the adjunction is defined  to have typical component $\eta_A:(\A, \equiv) \to \fame(\col(\A, \equiv))$
given by 
$
 \eta_A(a) = \setof{a}_\equiv \,.
$
It is easy to check that $\eta_A$ is a map in $\FAME$.  

\begin{thm}\label{thm:adjFAMEtoGES}
Suppose that $B= (B, \Con_B, \vdash_B) \in \GES$ and that $g: (\A, \equiv) \to (\iconf B, =)$ is a map in $\FAME$.  Then, there is a unique map $k: \col(\A, \equiv) \to B$ in $\GES$ such that the diagram
\[
\xymatrix{
(\A,\equiv) \ar[r]^{\eta_A}\ar[dr]_g&\fame(\col(\A, \equiv)) \ar@{..>}[d]^{\fame(k)}\\
& (\iconf B, =)
}
\]
commutes.
\end{thm}
\begin{proof}
The map $k: \col(\A, \equiv) \to B$ is given as the  function 
$
k(e) = g(a)\hbox{ where } e=\setof a_\equiv \,.
$
It is easily checked to be a map in $\GES$ and moreover to be the unique map from  $\col(\A, \equiv)$ to
$B$ making the above diagram commute.
 \end{proof}

Theorem~\ref{thm:adjFAMEtoGES} determines an adjunction:
\[\xymatrix{
\FAME\ar@/_/[rr]_{\col}^{\top} &&   \GES  \ar@/_/[ll]_{\fame} 
}\]
The construction $\col$ automatically extends from objects to maps; maps in $\FAME$ preserve equivalence so collapse to functions preserving equivalence classes.  
The counit of the adjunction has components $\eps_E: \col((\iconf E, =)) \to E$ which send singleton equivalence classes $\setof e$ to $e$.  The counit is an isomorphism at precisely those general event structures $E$ which are replete, so cuts down to a reflection from the subcategory of replete general event structures into equivalence families.

Composing  
\[\xymatrix{
\ESE \ar@/_/[rr]_{I}^{\top} &&  \FAME \ar@/_/[ll]_{\er}\ar@/_/[rr]_{\col}^{\top} &&   \GES  \ar@/_/[ll]_{\fame} 
}\]
we obtain a pseudo adjunction
\[
\xymatrix{
\ESE  \ar@/_/[rr]_{ }^{\top} &&  \GES \ar@/_/[ll]_{ }\,.
}
\]
  Its right adjoint constructs the {\em causal unfolding} 
of a general event structure.  

The composite pseudo adjunction from $\ESE$ to $\GES$ cuts down to a reflection, in  the sense that the counit is a natural isomorphism, when we restrict to the subcategory of $\GES$ where all general event structures are replete.  Then the right adjoint provides a pseudo functor embedding replete general event structures (and so families of configurations) in \ese's.

This concludes the construction of causal unfoldings of  (very general) equivalence-families, and, in particular, general event structures.  

We can ask for those extra axioms an \ese\, should satisfy in order that it arises from a general event structure;
an axiomatisation is given in Appendix~\ref{secn:gesasese}.

\section{Pullbacks  of \ese's?} 

A major motivation  has been to develop probabilistic strategies with parallel causes.  In the composition of strategies the constructions of (pseudo) pullback and hiding play a crucial role.  In composing strategies, essentially by playing them off against each other---the role of (pseudo) pullbacks, it is important to hide the moves associated with this interaction. 
Do  event structures with equivalence, $\ESE$, support these constructions?

 \Ese's do support  hiding. Let $(P, \leq, \Con_P, \eeq)$ be an \ese.  Let $V\subseteq P$ be a $\eeq$-closed subset of `visible' events. Define the {\em projection} of $P$ on $V$, to be $P{\mathbin\downarrow} V\eqdef (V, \leq_V, \Con_V, \eeq_V)$, where $v \leq_V v' \hbox{ iff } v\leq v' \ \&\ v,v'\in V$ and $X\in\Con_V \hbox{ iff }  X\in\Con\ \&\ X\subseteq V$ and $v \eeq_V v' \hbox{ iff } v\eeq v' \ \&\ v,v'\in V$. 

Hiding is associated with a factorisation of partial maps. Let $f$ be a partial map from $(P, \leq_P, \Con_P, \eeq_P)$ to $(Q, \leq_Q, \Con_Q, \eeq_Q)$. Letting $V\eqdef \set{e\in E}{ f(e) \hbox{ is defined}}$, the map $f$  factors into the composition
\[
\xymatrix{
P\ar[r]^{f_0}& P{\mathbin\downarrow} V \ar[r]^{f_1}& Q }
\]
of $f_0$, a partial map of \ese's taking $p\in P$ to itself if $p\in V$ and undefined otherwise, and $f_1$, a total map of \ese's acting like $f$ on $V$. We call $f_1$ the {\em defined part} of the partial map $f$.  Because $\eeq$-equivalent maps share the same domain of definition, $\eeq$-equivalent maps will determine the same projection and $\eeq$-equivalent defined parts. The factorisation is characterised to within isomorphism by the following  universal characterisation:  for any factorisation $\xymatrix{P\ar[r]^{g_0}& P_1 \ar[r]^{g_1}& Q }$ where $g_0$ is partial and $g_1$ is total there is a  (necessarily total) unique map $h: P{\mathbin\downarrow} V\to P_1$  such that we obtain the commuting diagram
\[
\xymatrix@R=12pt@C=20pt{
P\ar[r]^{f_0}\ar[dr]_{g_0}& P{\mathbin\downarrow} V\ar@{-->}[d]^h \ar[r]^{f_1}& Q \\
 & P_1\,.\ar[ur]_{g_1}& }
\]
 
By analogy with early work on prime event structures  and their representation by stable families~\cite{icalp82,evstrs} we might hope to obtain pullbacks and pseudo pullbacks in $\ESE$ via the adjunction to $\FAME$.  The category
$\FAME$ has pullbacks and pseudo pullbacks which are easy to construct---see Section~\ref{sec:constrns}.  

But unfortunately (pseudo) pullbacks in $\FAME$ don't provide us with  (pseudo) pullbacks in $\ESE$ because the right adjoint is only a pseudo functor:  in general it will only carry pseudo pullbacks to bipullbacks. While $\ESE$ does have bipullbacks (in which commutations and uniqueness are only up to the equivalence $\eeq$ on maps) it doesn't always have pseudo pullbacks or pullbacks---Appendix~\ref{app:pbsofeses}.  Whereas pseudo pullbacks and pullbacks are characterised up to isomorphism, bipullbacks are only characterised up to a weaker equivalence---that induced  on objects by the equivalence on maps.\footnote{Objects $P$ and $Q$ are equivalent iff there are two maps $f:P\to Q$,  $g:Q\to P$ s.t.~$gf\eeq \id_P$ and $fg\eeq \id_Q$.}   

We explore subcategories of $\ESE$, in particular w.r.t.~whether they support (pseudo) pullbacks.

\subsection{Subcategories of  \texorpdfstring{$\ESE$}{E-equiv}}\label{sec:coreflnsinESE}
\newcommand{\rest}{r}

Consider the following successively weaker axioms on an \ese~$(P, \Con, \leq, \eeq)$:
\begin{itemize}
\item[]{\bf 
Ax0.\ }
$\setof{p_1, p_2}\in\Con \ \&\ p_1\eeq p_2 \implies p_1= p_2\,.$
\item[]{\bf 
Ax1.\ }
$p_1, p_2\leq 
p \ \&\ p_1\eeq p_2 \implies p_1= p_2\,.$
\item[]{\bf 
Ax2.\ }
$p_1\leq p_2  \ \&\ p_1\eeq p_2 \implies p_1= p_2\,.$
\end{itemize}
Ax0 says that any two prime causes of disjunctive event are mutually exclusive.  Ax2 we have met as a consequence of a realisation being extremal (Lemma~\ref{lem:extperfect}(i)) so it will always hold of any image under the construction $\er$.  Ax1 forbids any prime cause from depending on two distinct prime causes of a common disjunctive event.  Example~\ref{ex:nonedc} shows Ax1 does not hold of all extremal realisations and can fail in an image under the construction $\er$.  Ax1
 enforces a form of atomicity on disjunctive events:  whereas several prime causes of a disjunctive event may appear in a configuration,   another event is not permitted to depend on, so detect, those several prime causes together.    

Restricting to the full subcategories  of $\ESE$  satisfying these axioms we obtain 
$\ESE^0$, $\ESE^1$ and $\ESE^2$
respectively.   The factorisation of maps we met for $\ESE$ is inherited by all the subcategories as their respective axioms are preserved by the projection operation.  So all the subcategories support hiding.  

The full inclusion functors
\[
\ESE^0\hookrightarrow \ESE^1 \hookrightarrow \ESE^2\hookrightarrow \ESE\]
all have right adjoints so forming a chain of coreflections.  Essentially the right adjoints work by restricting 
the structures to that part satisfying the stronger axiom. 
The adjunctions are enriched in the sense that the associated natural isomorphisms
preserve and reflect the equivalence $\eeq$ between maps (see Appendix~\ref{app:equiv-enriched}).

For example, 
$\ESE^0$ is the full subcategory of $\ESE$ in which objects $(P, \Con, \leq, \eeq)$ satisfy the strongest axiom 
 Ax0.
Consequently its maps are traditional maps of event structures which preserve equivalence.
The inclusion functor $\ESE^0\hookrightarrow \ESE$ has a right adjoint $\rest:\ESE\to \ESE^0$ 
taking $Q=(Q, \Con_Q, \leq_Q, \eeq_Q)$ to $(Q',\Con',\leq',\eeq')$ where
\begin{itemize}
\item
$Q'$ consists of all $q\in Q$ s.t.~
$q_1\not\eeq_Q q_2$ for 
all $q_1, q_2\leq_Q q$;
\item
$X\in\Con'$ iff $X\subseteq Q'$ and $X\in\Con_Q$ and  $q_1\not\eeq_Q q_2$ for all $q_1,q_2\in X$; 
\item
$\leq'$ and $\eeq'$ are the restrictions of $\leq_Q$ and $\eeq_Q$ to $Q'$.  
\end{itemize}
The adjunction being enriched means that the  isomorphism
\[
\ESE^0(P, \rest(Q)) \iso \ESE^0(P, Q)\,,
\]
between homsets, beyond being natural in $P\in \ESE^0$ and $Q\in\ESE$, 
preserves and reflects the equivalence $\eeq$ between maps.  
 
As a consequence we obtain an adjunction from $\ESE^0$ to $\GES$.
The universality of counit is only up to $\eeq$.\footnote{It was falsely claimed in \cite{evstrs} that the `inclusion' of the category of prime event structures in that of general event structures had a right adjoint. The adjunction from $\ESE^0$ to $\GES$ corrects that originally incorrect idea; though the repair 
 is at the cost of uniqueness up to $\eeq$.}

The most important  subcategory for us is $\ESE^1$. Amongst the subcategories of $\ESE$ it is the smallest  extension of prime event structures which supports parallel causes and hiding.   Its objects are called {\em event structures with disjunctive causes (\edc's)}~\cite{CSL17}.  The right adjoint to the inclusion
\[
\ESE^1 \hookrightarrow \ESE
\]
on objects simply restricts them to those events which satisfy Ax1.
In general, within $\ESE$ we lose the local injectivity property that we're used to seeing for maps of event structures.  However for $\ESE^1$ we recover local injectivity w.r.t.~prime configurations $[p]$. 
If $f:P\to Q$ is a map in $\ESE^1$, then
\[
p_1,p_2 \in [p] \ \& \ f(p_1) = f(p_2) \implies p_1 = p_2\,.
\]

In the composite adjunctions from 
$
\ESE^1$ to $\FAME$, and from 
$
\ESE^1$ to $\GES$,  the right adjoint has
  the effect of restricting to those extremal realisations within which Ax1 holds; recall that the  prime extremal realisations of an equivalence family $\A$ correspond to the configurations of $\er(\A)$.  Because such prime extremals are necessarily injective functions their carriers can be taken to be configurations of the equivalence family or general event structure of which they are realisations.  
Appendix~\ref{secn:gesasedc} provides an axiomatisation of those \edc's which arise from general event structures.

 As we shall see, $\ESE^1$ also has pullbacks and pseudo pullbacks.  It is within $\ESE^1$ that we have developed probabilistic distributed strategies with parallel causes~\cite{CSL17,ecsym-notes}. 
The coreflection from $\ESE^0$ to $\ESE^1$ is helpful in thinking about constructions like pullback and pseudo pullback in $\ESE^1$ as its right adjoint will preserve such limits.  
In the category $\ESE^0$,  maps coincide with the traditional maps of labelled event structures, regarding events as labelled by their equivalence classes.  
Constructions such as pullback are already very familiar in $\ESE^0$.  What changes in the corresponding constructions in $\ESE^1$ is the manner of dealing with consistency.  

 While not strictly a subcategory of $\ESE$, we should also mention the relationship with the category $\mathcal{E}$ of traditional event structures.  
There is an obvious `inclusion' functor from the category of event structures $\mathcal{E}$ to the category $\ESE^0$;  it takes an event structure to the same event structure but with the identity equivalence
adjoined.   Regarding $\ESE^0$ as a category, so dropping the enrichment by equivalence relations, the `inclusion' functor
\[
 \mathcal{E} \hookrightarrow \ESE^0
\]
has a right adjoint, \viz~the forgetful functor which simply drops the equivalence $\eeq$ from the \ese. The adjunction is  a coreflection because the inclusion functor is full.
Of course it is no longer the case that the adjunction is enriched: the natural bijection of the adjunction cannot respect the equivalence on maps. 

The adjunction
\[
 \mathcal{E} \hookrightarrow \ESE^1
\]
is obtained as the composite of the adjunctions from 
 $\mathcal{E}$ to $\ESE^0$ and $\ESE^0$ to $\ESE^1$ and is necessarily not enriched.
 Despite this the adjunction from $\mathcal{E}$ to $\ESE^1$ has been useful in relating strategies based on \edc's to strategies based on event structures.  Its right adjoint, the functor forgetting  equivalence, preserves all limits and especially pullbacks important in composing strategies. While this does not entail that composition of strategies is preserved by the forgetful functor---because the forgetful functor does not commute with hiding, it gives us a strong relationship, between composition of strategies before and after applying the forgetful functor~\cite{ecsym-notes}.

\subsection{Edc's and stable \ef's---a coreflection}\label{secn:stableefs}

The previous section positioned the category edc's $\ESE^1$ with respect to other subcategories of ese's. 
We rechristen the category $\ESE^1$ to $\EDC$ in view of its importance in the theory of probabilistic distributed strategies with parallel causes. 
The closeness of \edc's  to event structures suggests a generalisation of stable families to  aid with constructions such as product, pullback  and pseudo pullback in $\EDC$.\footnote{It is hard to define the product and pullback of prime event structures directly.  However the category of prime event structures is in coreflection with the  category of stable families where product and pullback are easily defined; for example, the product of event structures is then obtained as the image under the right adjoint of the product of their stable families of configurations~\cite{icalp82,evstrs}.}  

We are fortunate in that  the complicated pseudo adjunction between \ese's and \ef's restricts to a much simpler adjunction, in fact a coreflection, between \edc's and {\em stable} \ef's now defined. 

In an equivalence family $(\A, \eeq_A)$ say a configuration $x\in \A$ is {\em unambiguous} iff $\forall a_1, a_2\in x.\ a_1\eeq_A a_2 \implies a_1 = a_2\,$. An equivalence family $(\A, \eeq_A)$, with underlying set of events $A$, is {\em stable} iff it satisfies
\begin{center}\begin{tabular}{l}
$\forall x,y,z\in\A.\ \ x, y\subseteq z \ \&\ z $ is unambiguous $\,\Rightarrow\, x\cap y\in \A$\,, \hbox{ and } \\
$\forall a\in A, x\in\A.\ \, a\in x \,\Rightarrow  \exists z\in\A,\, z$ is unambiguous $\&\, a\in z\subseteq x\,$.
\end{tabular}\end{center}
Given the other axioms of an ef,  we can deduce the seemingly stronger property
\[
\emptyset\neq X \subseteq \A, z\in \A.\ (\forall x\in X.\ x\subseteq z) \  \&\ z \hbox{ is unambiguous } \implies \bigcap X \in\A\,
\]
of a stable \ef~$\A$. 

In effect, a stable equivalence family contains a stable subfamily of unambiguous configurations out of which all other configurations are obtainable as unions.  Local to any unambiguous configuration $x$ there is a partial order on its events $\leq_x$:  each  $a\in x$ determines a {\em prime configuration}
\[
[a]_x \eqdef \bigcap\set{y\in \A}{a\in y\subseteq x}\,,
\]
the minimum set of events on which $a$ depends within $x$;  taking $a\leq_x b$ iff $[a]_x\subseteq [b]_x$ defines causal dependency between  $a, b\in x$. Write $\SFAME$ for the subcategory of stable \ef's.

 The configurations of an \edc~with its equivalence are easily seen to form a stable \ef~providing a full and faithful embedding of   $\EDC$ in $\SFAME$.  The embedding has a right adjoint $\Pr$. It is built out of prime extremals but we can take advantage of the fact that in a stable \ef~unambiguous prime extremals have the simple form of prime configurations. From a stable \ef $(\A, \eeq_A)$ we produce an \edc~$\Pr(\A, \eeq_A) =_{{\rm def}} (P, \Con, \leq, \eeq)$ in which  $P$ comprises the prime configurations with  
\begin{itemize}
\item
$[a]_x \eeq [a']_{x'} \hbox{ iff } a\eeq_A a'\,,$
\item 
$Z \in \Con  \hbox{ iff }  Z \subseteq P \ \&\ \bigcup Z \in\F\,, \hbox{ and}$
\item
$p\leq p' \hbox{ iff }  p, p'\in P\ \&\ p\subseteq p'$\,.
\end{itemize}
The adjunction is enriched in the sense that its natural bijection preserves and reflects the equivalence on maps: \\
\centerline{$ \xymatrix{\EDC  \ar@/_/[rr]_{ }^{\top} &&  \SFAME \ar@/_/[ll]_{\Pr}}$}

Compare the definition above with that of $\Pr$ on   stable families.  The   significant difference is in the way that consistency is defined; in the construction on a stable ef the consistency is inherited not from the stable family of unambiguous configurations but from the ambient ef $\A$ in which configurations may not be unambiguous.

A stable equivalence family $\A$ contains a stable subfamily $\unamb(\A)$ of unambiguous configurations out of which all other configurations are obtainable as unions. There is an obvious `inclusion' functor from the category of stable families $\SFAM$ to $\SFAME$; it takes a stable family $\A$, with underlying set $A$, to the stable ef $(\A, \id_A)$.  Its has $\unamb$ as a right adjoint:
\[
\xymatrix{
\SFAM  \ar@/_/[rr]_{ }^{\top} &&  \SFAME \ar@/_/[ll]_{\unamb}\,.
}
\]
 As the `inclusion' functor from $\SFAM$ to $\SFAME$ is full the adjunction is a coreflection.  The adjunction is not enriched in the sense that its natural bijection ignores the equivalence on maps present in $\SFAME$.  As right adjoints preserve limits, the stable family of unambiguous configurations of the product, or pullback, of stable ef's is the product, respectively pullback,  in stable families of the unambiguous configurations of the components.

We can now obtain a (pseudo) pullback in  \edc's by first forming the (pseudo) pullback of the stable \ef's obtained as their configurations and then taking its image under the right adjoint $\Pr$.  
 
\subsection{Constructions}\label{sec:constrns}

We make use of the constructions of product, pullback and pseudo pullback of ef's; just as with prime event structures we cannot expect such constructions to be easily achieved directly on \ese's. 
 The constructions of product and pullback of ef's will reduce to product and pullback on families of configurations when we take the equivalences $\eeq$ to be the identity relation.  
On stable families they reduce to the product and pullback of stable families~\cite{icalp82,evstrs}.

The {\em product} of ef's is given as follows.
Let $\A$ and $\B$ be ef's with underlying sets $A$ and $B$.  Their product will have underlying set $A\times_* B$, the product   of $A$ and $B$ in sets with partial functions with projections $\pi_1$ to $A$ and $\pi_2$ to $B$.  We take $c\eeq c'$ in $A\times_* B$ iff   $\pi_1 c \eeq \pi_1 c'$, or both are undefined,  and $\pi_2 c \eeq \pi_2 c'$, or both are undefined.  
Define the configurations of the product by:
$x\in \A \times \B$ iff
\begin{itemize}
\item
$ x\subseteq A\times_* B $  such that
\item
$
\pi_1 x \in \A  \ \& \  \pi_2 x \in \B\,,
$
\item
$
\forall c, c' \in x. \   \pi_1 (c) \eeq_A \pi_1 (c') \hbox{ or } \pi_2 (c) \eeq_B \pi_2 (c')  \implies  c\eeq c'\,$
and
 \item
$ 
\forall c\in x \exists c_1, \cdots, c_n\in x.\ c_n =c \ \&\ 
\\
 \hbox{\ }\,  \quad \forall i \leq n.\   \pi_1\setof{c_1,\cdots, c_i}\in \A \ \&\ \pi_2\setof{c_1,\cdots, c_i}\in\B\,.
$
\end{itemize}

We obtain the product in stable ef's by restricting to those configurations of the product of the stable ef's which are unions of unambiguous configurations.  Notice that unambiguous configurations of the product of stable ef's are exactly the configurations in the product in stable families of the subfamilies of unambiguous configurations.

Restriction w.r.t.~sets of events which are closed under $\eeq$ and synchronised compositions are defined analogously to  before.  In particular we obtain pullbacks and bipullbacks as restrictions of the product.  

Pullbacks exist in general but we concentrate on pullbacks of total maps.
Let $f: \A \to \C$ and $g:\B\to\C$ be total maps of ef's. Assume $\A$ and $\B$ have underlying sets $A$ and $B$.  Define  $D\eqdef\set{(a,b)\in A\times B}{f(a) = g(b)}$ with projections $\pi_1$ and $\pi_2$ to the left and right components.  On $D$, take $d\eeq_D d'$ iff $\pi_1(d)\eeq_A \pi_1(d')$ and  $\pi_2(d)\eeq_B \pi_2(d')$.  Define a family of configurations of the {\em pullback} to consist of 
$x\in \D$ iff
\begin{itemize}
\item
$ x\subseteq D$  such that
\item
$
\pi_1 x \in \A  \ \& \  \pi_2 x \in \B\,,
$
and
 \item
$
\forall d\in x \exists d_1, \cdots, d_n\in x.\ d_n =d \ \&\ 
\\
\hbox{\ }\,  \quad \forall i \leq n.\   \pi_1\setof{d_1,\cdots, d_i}\in \A \ \&\ \pi_2\setof{d_1,\cdots, d_i}\in\B\,.
$
\end{itemize}

The pullback in stable ef's is again obtained by restricting to those configurations which are unions of unambiguous configurations.  The unambiguous configurations in the pullback of stable ef's are obtained as the pullback in stable families of the subfamilies of unambiguous configurations.

Given that maps are related by an equivalence relation it is sensible to broaden our constructions to pseudo pullbacks---the universal characterisation of pseudo pullback follows the concrete construction. 

Pseudo pullbacks   of total maps
$f: \A \to \C$ and $g:\B\to\C$   of ef's are obtained in a similar way to pullbacks. Assume $\A$ and $\B$ have underlying sets $A$ and $B$.  Define  $D\eqdef\set{(a,b)\in A\times B}{f(a) \eeq_C g(b)}$ with projections $\pi_1$ and $\pi_2$ to the left and right components.  On $D$, take $d\eeq_D d'$ iff $\pi_1(d)\eeq_A \pi_1(d')$ and  $\pi_2(d)\eeq_B \pi_2(d')$.  Define a family of configurations of the {\em pseudo pullback} to consist of 
$x\in \D$ iff
\begin{itemize}
\item
$ x\subseteq D$  such that
\item
$
\pi_1 x \in \A  \ \& \  \pi_2 x \in \B\,,
$
and
 \item
$
\forall d\in x \exists d_1, \cdots, d_n\in x.\ d_n =d \ \&\ 
\\
\hbox{\ }\,  \quad \forall i \leq n.\   \pi_1\setof{d_1,\cdots, d_i}\in \A \ \&\ \pi_2\setof{d_1,\cdots, d_i}\in\B\,.
$
\end{itemize}
When $\A$ and $\B$ are stable ef's we obtain their pseudo pullback by restricting to those configurations obtained as the union of unambiguous configurations.   
 
Recall the universal property of a pseudo pullback of $f:\A\to \C$ and $g:\B\to \C$
(in this simple case).  A pseudo pullback
comprises two maps $\pi_1:\D\to \A$ and $\pi_2:\D\to \B$ such that
$f\pi_1 \eeq g\pi_2$ with the universal property that given any two maps $p_1:\D'\to \A$ and $p_2:\D'\to \B$ such that 
$f p_1 \eeq g p_2$ there is a unique map $h:\D'\to \D$   such that $p_1= \pi_1h$ and $p_2= \pi_2h$:
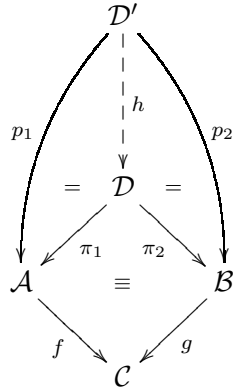
\begin{displaymath}
    \xymatrix@=2em{
        {} & \D' \ar@/^1pc/[dddr]^{p_2} \ar@/_1pc/[dddl]_{p_1}
\ar@{-->}[dd]^{h} & {} \\
        \\
                {} & \D\ar@{}[r]|=\ar@{}[l]|= \ar[dr]_{\pi_2} \ar[dl]^{\pi_1} & {} \\
                \A \ar[dr]_{f} \ar@{}[rr]|\eeq& {} & \B  \ar[dl]^{g} \\
                {} & \C & {} \\
}
\end{displaymath}
Pseudo pullbacks are defined up to isomorphism.  Pseudo pullbacks coincide with pullbacks when the maps involved have an event structure as their common codomain. 
 
Fortunately we do have both pullbacks and pseudo pullbacks in the subcategory $\ESE^1$.  
The constructions of pullbacks and pseudo pullbacks in $\ESE^1$ can by-pass the complicated $\er$ construction and be done via the corresponding constructions in $\SFAME$ in the manner  familiar  from event structures and stable families.  This is because we have an adjunction from $\ESE^1$ to $\SFAME$ and moreover an adjunction which is enriched with respect the equivalence on homsets.  So, for example, to form the (pseudo) pullback of \ese's in $\ESE^1$ we regard their configurations  as stable ef's, form the (pseudo) pullback in $\SFAME$ and take the image under the right adjoint $\Pr$.  Each stable ef includes a subfamily of unambiguous configurations and it is fortunate indeed that \eg~the subfamily of unambiguous configurations of the pullback of stable ef's $f:\A\to \C$ and $g:\B\to\C$ 
 is got as the pullback in stable families of $f$ and $g$ between the subfamilies of unambiguous configurations.

\section{Conclusion}

This completes our exploration of extensions of event structures to support disjunctive and especially parallel causes. We summarise the models and adjunctions we have met in a figure.

	\begin{center}
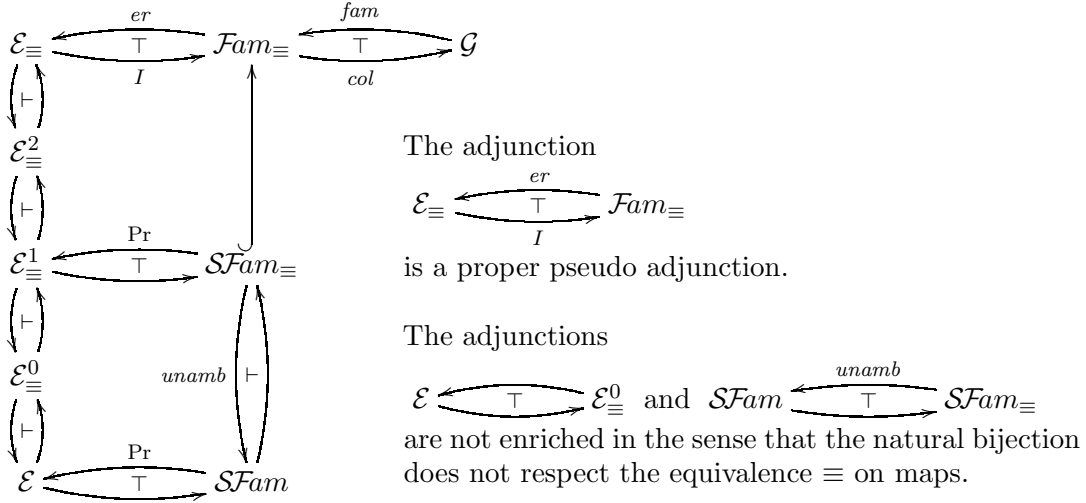
	\begin{tabular}{p{5cm}p{9cm}}
	$\xymatrix{ \ESE \ar@/_/[rr]_{I}^{\top} \ar@/_/[d]_{ }&& \FAME \ar@/_/[ll]_{\er}\ar@/_/[rr]_{\col}^{\top} && \GES \ar@/_/[ll]_{\fame} \\ \ESE^2\ar@/_/[u]_{ }^{\vdash}\ar@/_/[d]_{ }&&&&\\ \ESE^1\ar@/_/[rr]_{ }^{\top} \ar@/_/[d]_{ }\ar@/_/[u]_{ }^{\vdash}&& \SFAME \ar@/_/[ll]_{\Pr} \ar@{^{(}->}[uu]\ar@/_/[dd]_{\unamb}\\ \ESE^0\ar@/_/[d]_{ }\ar@/_/[u]_{ }^{\vdash}&&&&\\ {\mathcal{E}}\ar@/_/[u]_{ }^{\vdash}\ar@/_/[rr]_{ }^{\top}&&\SFAM\ar@/_/[uu]_{ }^{\vdash}\ar@/_/[ll]_{\Pr} }$ &
	~\newline\newline\newline The adjunction\newline $\xymatrix{ \ESE \ar@/_/[rr]_{I}^{\top} && \FAME \ar@/_/[ll]_{\er}}$\newline is a proper pseudo adjunction. \newline\newline
	The adjunctions \newline $\xymatrix{ {\mathcal{E}} \ar@/_/[rr]_{ }^{\top} && \ESE^0 \ar@/_/[ll]_{} }$  and $\xymatrix{ \SFAM \ar@/_/[rr]_{ }^{\top} && \SFAME \ar@/_/[ll]_{\unamb} }$ \newline are not enriched in the sense that the natural bijection does not respect the equivalence $\eeq$ on maps.
	\end{tabular}
\end{center}

The adjunctions of the figure bridge between the two ``classical'' models of prime event structures $\mathcal{E}$ and general event structures $\mathcal{G}$; the former appropriate when each event has a unique cause, the latter permitting causes in parallel. The other ``classical'' model of stable families $\SFAM$ supports exclusive disjunctive causes; an event may be enabled in several, though incompatible, ways. Equivalence families $\FAME$ and its subcategories provide us with constructions of product and (pseudo) pullback. All the categories to the left of the figure, subcategories of \ese's $\ESE$, have the partial-total factorisation property used to support an operation of hiding. The category $\ESE^1$ (a.k.a.~$\EDC$) of \edc's distinguishes itself in also having pullbacks and pseudo pullbacks. It thus answered the original motivation for our search, to find a model in which to develop strategies with parallel causes. The category $\ESE^1$ supports hiding and probability, has pullbacks, pseudo pullbacks and through this leads to a robust definition of probabilistic strategies with parallel causes~\cite{CSL17,ecsym-notes}. The tools described here have been essential in carrying out that programme. While parallel causes are ubiquitous, their more formal treatment has been rather sparse. The techniques of this paper should be relevant wherever causal models allowing parallel causes are in use. \\

\noindent{\bf Acknowledgments}{ Many thanks to the anonymous referees. Thanks to  Simon Castellan, Pierre Clairambault, Ioana Cristescu, Mai Gehrke, Jonathan Hayman, Tamas Kispeter, Jean Krivine, Martin Hyland and Daniele Varacca for discussions, advice and encouragement; to ENS Paris for supporting Marc de Visme's internship; and to the ERC for Advanced Grant ECSYM.}

\bibliographystyle{alphaurl}
\bibliography{CALCO19-specialissue}

\newcommand{\etalchar}[1]{$^{#1}$}
\begin{thebibliography}{DFF{\etalchar{+}}12}

\bibitem[CCW17]{fscd17}
Simon Castellan, Pierre Clairambault, and Glynn Winskel.
\newblock Observably deterministic concurrent strategies and intensional full
  abstraction for parallel-or.
\newblock In {\em {FSCD}}, 2017.
\newblock \href {https://doi.org/10.4230/LIPIcs.FSCD.2017.12}
  {\path{doi:10.4230/LIPIcs.FSCD.2017.12}}.

\bibitem[CKV15]{CKV}
Ioana Cristescu, Jean Krivine, and Daniele Varacca.
\newblock Rigid families for {CCS} and the {pi-Calculus}.
\newblock In {\em ICTAC}, 2015.
\newblock \href {https://doi.org/10.1007/978-3-319-25150-9_14}
  {\path{doi:10.1007/978-3-319-25150-9_14}}.

\bibitem[Cri15]{ioana}
Ioana Cristescu.
\newblock {\em Operational and denotational semantics for the reversible
  pi-calculus}.
\newblock PhD thesis, PPS, Universit\'e Paris Diderot, 2015.

\bibitem[DFF{\etalchar{+}}12]{kappa}
Vincent Danos, Jerome Feret, Walter Fontana, Russell Harmer, Jonathan Hayman,
  Jean Krivine, Chris Thompson-Walsh, and Glynn Winskel.
\newblock {Graphs, Rewriting and Pathway Reconstruction for Rule-Based Models}.
\newblock In {\em FSTTCS}, 2012.
\newblock \href {https://doi.org/10.4230/LIPIcs.FSTTCS.2012.276}
  {\path{doi:10.4230/LIPIcs.FSTTCS.2012.276}}.

\bibitem[dV15]{MarcRep}
Marc de~Visme.
\newblock Cambridge internship report, {ENS Paris}.
\newblock {Available} from {Glynn Winskel'}s homepage
  http://www.cl.cam.ac.uk/$\sim$gw104/mdv-report.pdf, 2015.

\bibitem[dVW17]{CSL17}
Marc de~Visme and Glynn Winskel.
\newblock Strategies with parallel causes.
\newblock In {\em {CSL}}, 2017.
\newblock \href {https://doi.org/10.4230/LIPIcs.CSL.2017.41}
  {\path{doi:10.4230/LIPIcs.CSL.2017.41}}.

\bibitem[Kel82]{kelly}
G.~M. Kelly.
\newblock {\em Basic concepts of enriched category theory}.
\newblock CUP, 1982.
\newblock \href {https://doi.org/10.1112/BLMS/15.1.96}
  {\path{doi:10.1112/BLMS/15.1.96}}.

\bibitem[KP14]{Equiv-cat}
Yoshiki Kinoshita and John Power.
\newblock Category theoretic structure of setoids.
\newblock {\em TCS}, 546, 2014.
\newblock \href {https://doi.org/10.1016/j.tcs.2014.03.006}
  {\path{doi:10.1016/j.tcs.2014.03.006}}.

\bibitem[NPW81]{NPW}
Mogens Nielsen, Gordon Plotkin, and Glynn Winskel.
\newblock Petri nets, event structures and domains.
\newblock {\em TCS}, 1981.
\newblock \href {https://doi.org/10.1016/0304-3975(81)90112-2}
  {\path{doi:10.1016/0304-3975(81)90112-2}}.

\bibitem[Pea13]{Pearl}
Judea Pearl.
\newblock {\em Causality}.
\newblock CUP, 2013.
\newblock \href {https://doi.org/10.1111/J.1751-5823.2011.00149_16.X}
  {\path{doi:10.1111/J.1751-5823.2011.00149_16.X}}.

\bibitem[Pow98]{Power2-cats}
John Power.
\newblock 2-categories.
\newblock {\em BRICS Notes Series NS-98-7}, 1998.

\bibitem[Win80]{thesis}
Glynn Winskel.
\newblock {\em Events in computation}.
\newblock {Edinburgh University}, 1980.
\newblock PhD thesis, Edinburgh.

\bibitem[Win82]{icalp82}
Glynn Winskel.
\newblock Event structure semantics for {CCS} and related languages.
\newblock In {\em ICALP}. Full version available as University of Aarhus, CS
  Technical Report 1983, and from
  http://www.cl.cam.ac.uk/$\sim$gw104/eventStructures82.pdf, 1982.
\newblock \href {https://doi.org/10.1007/BFb0012800}
  {\path{doi:10.1007/BFb0012800}}.

\bibitem[Win86]{evstrs}
Glynn Winskel.
\newblock Event structures.
\newblock In {\em Advances in Petri Nets}, LNCS 255, 1986.
\newblock \href {https://doi.org/10.1007/3-540-18086-9}
  {\path{doi:10.1007/3-540-18086-9}}.

\bibitem[Win16]{ecsym-notes}
Glynn Winskel.
\newblock {\em ECSYM Notes: Event Structures, Stable Families and Concurrent
  Games}.
\newblock http://www.cl.cam.ac.uk/$\sim$gw104/ecsym-notes.pdf, 2016.

\bibitem[WN95]{WN}
Glynn Winskel and Mogens Nielsen.
\newblock Models for concurrency.
\newblock 1995.
\newblock \href {https://doi.org/10.1016/S0304-3975(96)80710-9}
  {\path{doi:10.1016/S0304-3975(96)80710-9}}.

\end{thebibliography}

\appendix

\section{{\bf Equiv}-enriched categories}\label{app:equiv-enriched}

Here we explain in more detail what we mean when we say ``enriched in the category of sets with equivalence relations'' and employ terms such as ``enriched adjunction,'' ``pseudo adjunction'' and ``pseudo pullback.''  The classic text on enriched categories is~\cite{kelly}, but for this paper the articles~\cite{Equiv-cat} and~\cite{Power2-cats} provide short, accessible introductions to the notions we use from Equiv-enriched categories and 2-categories, respectively. 

 $\Equiv$ is the category of {\em equivalence relations}.  Its objects are $(A,\eeq_A)$ comprising a set $A$ 
and an equivalence relation $\eeq_A$ on it.  Its maps 
$f:(A, \eeq_A)\to (B,\eeq_B)$ are total functions $f:A\to B$ which preserve equivalence.

We shall use some basic notions from enriched category theory~\cite{kelly}.  We shall be concerned with categories enriched in $\Equiv$, called  $\Equiv$-enriched categories, in which the homsets possess the structure of equivalence relations, respected by composition~\cite{Equiv-cat}. This is the sense in which we say categories are enriched in (the category of) equivalence relations. We similarly borrow the concept of  an $\Equiv$-enriched functor between $\Equiv$-enriched categories for a functor which preserves equivalence in acting on homsets.   An $\Equiv$-enriched adjunction  is a usual adjunction in which  the natural bijection of the adjunction preserves and reflects equivalence.

Because an object in $\Equiv$ can be regarded as a (very simple) category, we can regard  $\Equiv$-enriched categories as (very simple) 2-categories to which notions from 2-categories apply~\cite{Power2-cats}.  

A {\em pseudo functor} between $\Equiv$-enriched categories is like a functor but the usual laws only need hold up to equivalence. A {\em pseudo adjunction} (or biadjunction) between 2-categories permits a weakening of the usual natural isomorphism between homsets, now also categories,  to a natural equivalence of categories. In the special case of a pseudo adjunction between $\Equiv$-enriched categories the equivalence of homset categories amounts to a pair of $\eeq$-preserving functions  whose compositions are $\eeq$-equivalent to the identity function. With traditional adjunctions, by specifying the action of one adjoint solely on objects, we determine it as a functor; with pseudo adjunctions we can only determine it as a pseudo functor---in general a pseudo adjunction relates two pseudo functors. Pseudo adjunctions compose in the expected way. An $\Equiv$-enriched adjunction is a special case of a 2-adjunction between 2-categories and a very special case of pseudo adjunction. In Section~\ref{sec:adjesetoges} we compose an $\Equiv$-enriched adjunction with a pseudo adjunction to obtain a new pseudo adjunction.  

Similarly we can specialise the notions pseudo pullbacks and bipullbacks from 2-categories  to $\Equiv$-enriched categories which is highly relevant to the companion paper~\cite{CSL17} in which we use pullbacks and pseudo pullbacks to compose strategies with parallel causes.  Let $f: A\to C$  and $g:B\to C$ be two maps in an $\Equiv$-enriched category.  A {\em pseudo pullback} of $f$ and $g$ is an object $D$ and maps $p:D\to A$ and $q:D\to B$ such that $f\circ  p \eeq g\circ  q $ which satisfy the further property that for any  $D'$ and maps $p':D'\to A$ and $q':D'\to B$ such that $f\circ  p' \eeq g\circ  q' $, there is a unique map $h:D'\to D$ such that $p' = p\circ  h$ and $q' = q\circ  h$; note the insistence on the last two equalities, rather than just equivalences. There is an obvious weakening of pseudo pullbacks to the situation in which the uniqueness is replaced by uniqueness up to $\eeq$ and the equalities by $\eeq$---these are simple special cases of bilimits called  {\em bipullbacks}. 

Right adjoints in a 2-adjunction preserve pseudo pullbacks whereas right adjoints in a pseudo adjunction are only assured to preserve bipullbacks.

\section{On (pseudo) pullbacks of ese's}\label{app:pbsofeses}

We show that the enriched category of \ese's $\ESE$ does not always have pullbacks and pseudo pullbacks of maps $f:A\to C$ and $g:B\to C$,  the reason why we use the subcategory $\EDC$, which does, as a foundation on which to develop strategies with parallel causes.  It suffices to exhibit the lack of pullbacks when $C$ is an (ese of an) event structure as then pullbacks and pseudo pullbacks coincide.  Take $A$, $B$, $C$ as below, with the  obvious maps $f:A\to C$ and $g:B\to C$ (given by the lettering).  In fact, $A$ and $B$ are edc's.
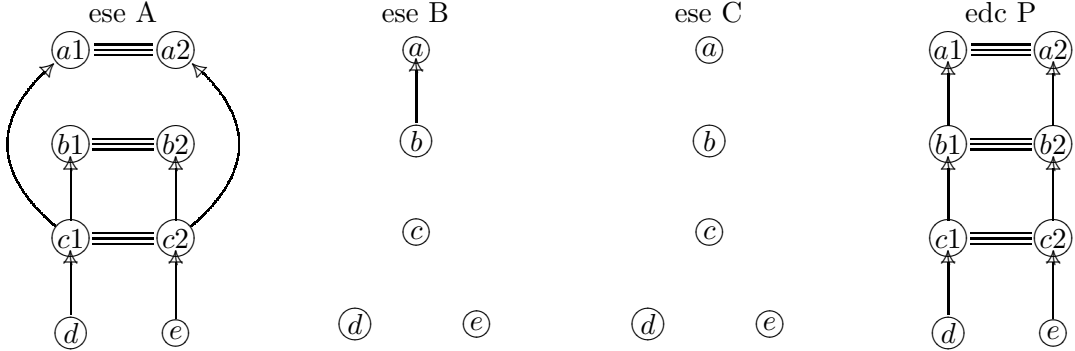
\begin{center}
\begin{tabular}{ccccccc}
\ese~A & & \ese~B & & \ese~C & & \edc~P
\\$
\xymatrix@C=2em@R=2em{
\ve{a1} \ar@3{-}[r] & \ve{a2} \\
\ve{b1} \ar@3{-}[r] & \ve{b2} \\
\ve{c1} \ar@3{-}[r] \ar@{|>}[u] \ar@{|>}@/^2pc/[uu] & \ve{c2} \ar@{|>}[u] \ar@{|>}@/_2pc/[uu] \\
\ve{d} \ar@{|>}[u] & \ve{e} \ar@{|>}[u]}
$&$ \qquad \quad $&$
\xymatrix@C=1em@R=2em{
{} & \ve{a} & {} \\
{} & \ve{b} \ar@{|>}[u] & {} \\
{} & \ve{c} & {} \\
\ve{d} & {} & \ve{e}}
$&$ \qquad \quad $&$
\xymatrix@C=1em@R=2em{
{} & \ve{a} & {} \\
{} & \ve{b} & {} \\
{} & \ve{c} & {} \\
\ve{d} & {} & \ve{e}}
$&$ \qquad \quad $&$
\xymatrix@C=2em@R=2em{
\ve{a1} \ar@3{-}[r] & \ve{a2} \\
\ve{b1} \ar@3{-}[r] \ar@{|>}[u] & \ve{b2} \ar@{|>}[u] \\
\ve{c1} \ar@3{-}[r] \ar@{|>}[u] & \ve{c2} \ar@{|>}[u] \\
\ve{d} \ar@{|>}[u] & \ve{e} \ar@{|>}[u]}
$\\
\end{tabular}
\end{center}

The pullback in \edc's $\EDC$ does exist and is given by $P$ with the obvious projection maps. However this is not a pullback in $\ESE$.  Consider the \ese~D with the obvious total maps to $A$ and $B$; they  form a commuting square with $f$ and $g$. This cannot factor through $P$:   event   $b2$ has to be mapped to $b2$ in $P$, but then $a1$ cannot be mapped to $a1$ (it wouldn't yield  a map) nor to $a2$ (it would violate commutation required  of a pullback). 
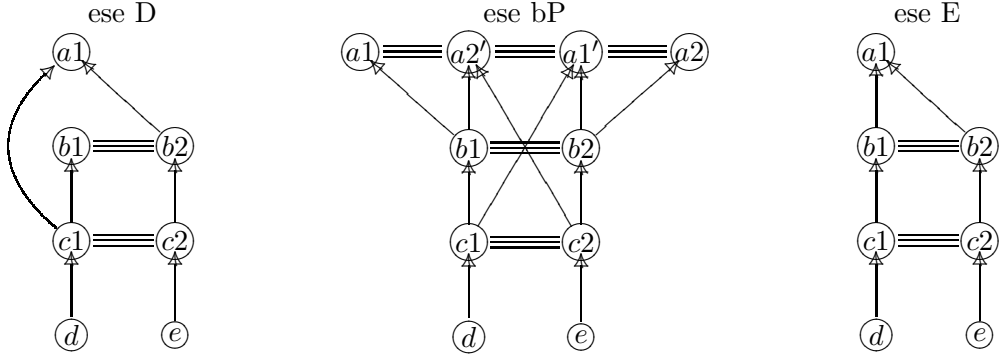
\begin{center}
\begin{tabular}{ccccc}
\ese~D & & \ese~bP & & \ese~E 
\\$
\xymatrix@C=2em@R=2em{
\ve{a1} & {} \\
\ve{b1} \ar@3{-}[r] & \ve{b2} \ar@{|>}[lu] \\
\ve{c1} \ar@3{-}[r] \ar@{|>}@/^2pc/[uu] \ar@{|>}[u] & \ve{c2} \ar@{|>}[u] \\
\ve{d} \ar@{|>}[u] & \ve{e} \ar@{|>}[u]}
$&$ \qquad \quad $&$ 
\xymatrix@=2em{
\ve{a1} \ar@3{-}[r] & \ve{a2'} \ar@3{-}[r] & \ve{a1'} \ar@3{-}[r] & \ve{a2} \\
{} & \ve{b1} \ar@3{-}[r] \ar@{|>}[ul] \ar@{|>}[u] & \ve{b2} \ar@{|>}[u] \ar@{|>}[ur] & {} \\
{} & \ve{c1} \ar@3{-}[r] \ar@{|>}[u] \ar@{|>}[uur] & \ve{c2} \ar@{|>}[u] \ar@{|>}[uul] & {} \\
{} & \ve{d} \ar@{|>}[u] & \ve{e} \ar@{|>}[u] & {}}
$&$ \qquad \quad $&$
\xymatrix@C=2em@R=2em{
\ve{a1} & {} \\
\ve{b1} \ar@3{-}[r] \ar@{|>}[u] & \ve{b2} \ar@{|>}[ul] \\
\ve{c1} \ar@3{-}[r] \ar@{|>}[u] & \ve{c2} \ar@{|>}[u] \\
\ve{d} \ar@{|>}[u] & \ve{e} \ar@{|>}[u]}
$\\
\end{tabular}
\end{center}

There is a bipullback $bP$ got by applying the pseudo functor $\er$ to the pullback in \ef's. But this is not a pullback because  in the \ese~$E$ the required mediating map is not unique in that  $a1$ can go to either $a1$ or $a1'$. In fact, there is no pullback of $f$ and $g$.  To show this we use the additional \ese~F. 
\begin{center}
\begin{tabular}{p{11cm} p{3cm}}
Suppose $Q$ with projection maps to $A$ and $B$ were a pullback of $f$ and $g$ in $\ESE$. Consider the three \ese's  $D$, $E$ and $F$ with their obvious maps to $A$ and $B$; in each case they form a commuting square with $f$ and $g$. There are three unique maps $h_D : D \to Q$, $h_E : E \to Q$, and $h_F : F \to Q$ such that the corresponding pullback diagrams commute. We remark that there are also  obvious maps $k_D : E \to D$ and   $k_F : E \to F$ (given by the lettering) which commute with the maps to the components $A$ and $B$. By uniqueness, we have $h_D \circ k_D = h_E = h_F \circ k_F$, so we have $h_D(a1) = h_F(a1)$. From the definition of the maps, the event $h_D(a1) = h_F(a1)$ has at most one
$\leq$-predecessor in $Q$  which is sent to $b$ in $C$ (as $D$ only has one). Because of the projection to $B$, it has at least one (as $B$ has one). So the event $h_D(a1) = h_F(a1)$ has exactly one predecessor which is sent to $b$. From the definition of maps, this event is $h_D(b2)$ which equals $h_F(b1)$. But $h_D(b2)$ cannot equal $h_F(b1)$ as they go to two different events of $A$ ---a contradiction.
&
\centering \ese~F \newline \newline
$
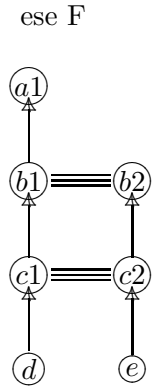
\xymatrix@C=2em@R=2em{
\ve{a1} & {} \\
\ve{b1} \ar@3{-}[r] \ar@{|>}[u] & \ve{b2}\\
\ve{c1} \ar@3{-}[r] \ar@{|>}[u] & \ve{c2} \ar@{|>}[u] \\
\ve{d} \ar@{|>}[u] & \ve{e} \ar@{|>}[u]
}$ \\
\end{tabular}
\end{center}
Hence there can be no pullback of $f$ and $g$ in $\ESE$. (By adding   intermediary events, we would encounter essentially the same example in the composition, before hiding, of strategies if they were to be developed within the broader category of \ese's.)

\section{General event structures as \ese's}\label{secn:gesasese}

The pseudo adjunction
\[
\xymatrix{
\ESE  \ar@/_/[rr]_{ }^{\top} &&  \GES \ar@/_/[ll]_{ }\,.
}
\]
 cuts down to a reflection, in which the counit is a natural isomorphism, when we restrict to the subcategory of $\GES$ where all general event structures are replete.  The right adjoint provides a full and faithful embedding of replete general event structures (and so families of configurations) in \ese's.   

We can ask on what subcategory of $\ESE$ the pseudo adjunction further cuts down to a pseudo equivalence  with the category of replete general event structures.    We do this by  characterising those \ese's 
which are obtained to within isomorphism as images of replete general event structures under the right adjoint, or equivalently as images of families of configurations.    

The characterising axioms on an \ese~$(P,\leq, \Con,\eeq)$ are:

\begin{itemize}[align=left]
\item[(A)]
For $X$ a finite  down-closed subset of $P$,\\
 $X\eeq y \ \&\ y \in \conf P \implies X\in\conf P$\,;
\item[(B)]
For $p, q\in P$, \ 
$[p) = [q)\ \&\ p \equiv q \implies p = q$\,;
\item[(C)] 
For $X$ a   down-closed subset of $P$ and $p\eeq q$,\\
$X\subseteq [p) \ \&\ [q)_\eeq \subseteq X_\eeq \implies X= [p)$\,; 
\item[(D)]
For $x\in\conf P$ and $t\in P$,\\
$x\cup [t] \in\conf P \ \&\ (x\cup[t])_\eeq = x_\eeq \cup \setof{\setof t_\eeq} \implies
\exists p\in P.\  p\eeq t \ \&\ x\cup\setof p\in\conf P$\,.
\end{itemize}

\noindent
In writing the axioms we have used expressions such as $X\eeq Y$, for subsets $X$ and $Y$ of $P$,  to mean for any $p\in X$ there is $q\in Y$ with $p\eeq q$ and {\it vice versa};  and $X_\eeq$ to stand for the set of $\eeq$-equivalence classes $\set{\setof p_\eeq}{p\in X}$; so $X\eeq Y$ iff $X_\eeq = Y_\eeq$.

Axiom (D) may be replaced by 
\begin{itemize}[align=left]
\item[(D$'$)]
For $x,y\in\conf P$ and $t\in P$,\\
$x\longcov t \ \&\ x\eeq y  \implies
\exists p\in P.\ p\eeq t \ \&\ x\cup\setof p\in\conf P$\,.
\end{itemize} 
Assume (D) and, for $x,y\in\conf P$, that $x\longcov t$ and $x\eeq y$.  Then, by (A),  $y\cup [t] \in\conf P$ as $y\cup [t]\eeq x\cup \setof t$, clearly consistent; whence $y\cup \setof p\in \conf P$ for some $p$ by (D). Conversely, assuming (D$'$) and $x\cup [t] \in\conf P$ and $(x\cup[t])_\eeq = x_\eeq \cup \setof{\setof t_\eeq}$, in the case where $t\notin x$ we obtain $x\cup [t) \longcov t$ and $x\cup [t) \eeq x$; whence $x \cup \setof p\in \conf P$ for some $p$ by (D$'$).  This shows $(D)$ follows from (D$'$) in the case when $t\notin x$; in the case when $t\in x$, axiom (D) is obvious.  

\begin{thm}\label{thm:ges-as-ese}
Let $P\in\ESE$.  Then, 
$P\iso \er(\A)$ for some equivalence family $\A$  
 iff $P$ satisfies axioms (A), (B), (C) and (D).
\end{thm}
\begin{proof} 
We show axioms (A), (B), (C), (D) hold of any 
\ese~$P = \er(\A)$, constructed from a 
family of configurations $\A$.  
We obtain $P$ satisfies axiom (A) 
from the way the consistency of  $\er(\A)$ is defined:  if $X\eeq y$, with $y$ a configuration, $X$ inherits consistency from $y$ ensuring that $X$, assumed down-closed, is a configuration. 
If $[p) = [q)$ and
$p \equiv q$, then $p$ and $q$ correspond to the same extremal realisation with top,
so are equal---ensuring (B) holds of $P$.
We obtain (C) via  Lemma~\ref{lem:charnofextremals}(i), as $[p]$ corresponds to an extremal with top $p$. 
Given the correspondence between configurations of $P$ and extremal realisations, axiom (D) expresses an obvious extension property of extremal realisations.

Conversely, we now show that if an \ese~$P=(P,\Con, \leq, \eeq)$ satisfies  (A), (B), (C), (D) then  there is an isomorphism  
  \[\eta_P:P \iso \er(\A
)\]  
if 
we take the family of configurations so
\[\A = \iconf{\col(\iconf P,\eeq)}\,.\] 
Recall, from
Proposition~\ref{prop:domianofrealzns}, that
the configurations of $\er(\A)$ correspond to extremal realisations of  $\col(\iconf P,\eeq)$. 

Before we define the map $\eta_P$ we remark that a configuration  $x$ of $P$  determines an extremal realisation of $\col(\iconf P,\eeq)$: the realisation has carrier $x$ with order inherited from $P$ and map taking $p\in x$ to the equivalence class $\setof p_\eeq$.  Axioms (B) and (C) ensure that this realisation is extremal, via Lemma~\ref{lem:charnofextremals}. 

It follows from the remark that we define a map $\eta_P:P \to \er(\A)$ by sending $p\in P$ to the realisation with carrier 
$[p]$, ordered as in $P$, and function  $[p]\to P_\eeq$ taking elements to their equivalence classes. 
The injectivity of $\eta_P$ follows from (B). Moreover $\eta_P$ reflects consistency because of axiom (A).  We now only require its surjectivity to ensure $\eta_P$ is an isomorphism.

We use (D) in showing that $\eta_P$ is surjective.  We show by induction on $n\in\omega$ that all extremal realisations with top of $\col(P)$ of depth less than $n$ are in the image of $\eta_P$.  (Recall the depth of an event in an event structure is the length of a longest $\leq$-chain up to it; we take the depth of an event structure to be the maximum depth of its events.)  Because $\eta_P$ reflects consistency the induction hypothesis entails that all extremal realisations  of depth less than $n$ are (up to isomorphism) in the image under $\eta_P$ of configurations of $P$.  

Let $(R, \leq_R)$ of depth $n$ with $\rho:R\to \col(P)$ be an extremal realisation with top $r$, so $R=[r]_R$. 
Then its restriction $\rho': [r)_R \to \col(P)$ is an extremal realisation of lesser depth.  By induction there is $x'\in\conf P$ and an isomorphism of realisations $\theta':  \rho' \iso \eta_P x'$.  Write $y\eqdef \rho' [r)_R$, $z\eqdef \rho [r]_R$.  Then $y, z\in \conf{\col(P)}$ and $y\longcov e z$ for some $e\in P_{\eeq}$.  From the definition of $\col(P)$, it follows fairly directly that there is some $t\in P$ s.t.~$\setof t_{\eeq} = e$ and $[t)_\eeq \subseteq y$.  As $\eta_P$ reflects consistency, $x'\cup [t] \in\conf P$.  We have
\[
(x'\cup [t])_\eeq = x'_\eeq \cup \setof{\setof t_\eeq} = z\,.
\]
By (D) there is some $p\in P$ s.t.~$p\eeq t$ and $x'\cup\setof p\in\conf P$.  The configuration $x\eqdef x'\cup\setof p$ with order inherited from $P$ and map taking $p'\in x $ to $\setof{p'}_\eeq$ is the realisation $\eta_P x$.
Let $\theta$ be the function $\theta:R\to x $ extending $\theta'$ s.t.~$\theta(r) = p$.  Then 
$\theta: \rho \succeq \eta_P x$ is a map of realisations.  But $\rho$ is extremal ensuring $\theta:  \rho \iso \eta_P x$, and that $\eta_P$ is surjective. 
 \end{proof}
 
 \begin{cor}
 The pseudo adjunction from $\ESE$ to $\GES$ cuts down to a pseudo equivalence of categories between the subcategory of $\ESE$ satisfying axioms (A), (B), (C), (D) and the subcategory of $\GES$ comprising the replete general event structures.
  \end{cor}

\subsection{General event structures as \edc's}\label{secn:gesasedc}
The composite
\[\xymatrix{
\ESE^1\ar@/_/[rr]_{ }^{\top} &&  \ar@/_/[ll]_{} \ESE \ar@/_/[rr]_{I}^{\top} &&  \FAME \ar@/_/[ll]_{\er}\ar@/_/[rr]_{\col}^{\top} &&   \GES  \ar@/_/[ll]_{\fame}
}
\]
forms pseudo adjunction from \edc's to general event structures; we call its composite right adjoint $\hbox{{\it \edc}}$. 
 As above we can cut this down to a pseudo equivalence between the  subcategory $\ESE^1$ of edc's and replete general event structures via a slight modification of Axiom (D) on \ese's~$(P,\leq, \Con,\eeq)$:

\begin{itemize}[align=left]
\item[(D$_1$)]
For $x$ an unambiguous configuration in $\conf P$ and $t\in P$,\\
$x\cup [t] \in\conf P \ \&\ (x\cup[t])_\eeq = x_\eeq \cup \setof{\setof t_\eeq} \implies
\exists p\in P.\  p\eeq t \ \&\ x\cup\setof p\in\conf P$\,.
\end{itemize}
Axiom (D$_1$) may be replaced by 
\begin{itemize}[align=left]
\item[(D$_1'$)]
For $x,y\in\conf P$, with $y$ unambiguous, and $t\in P$,\\
$x\longcov t \ \&\ x\eeq y  \implies
\exists p\in P.\ p\eeq t \ \&\ x\cup\setof p\in\conf P$\,.
\end{itemize}

\begin{thm}\label{thm:ges-as-edc}
Let $P$ be an edc.  Then, 
$P\iso \edc(G)$ for some general event structure $G$  
 iff $P$ satisfies axioms (A), (B), (C) and (D$_1$).  
\end{thm}
\begin{proof} The proof is a slight refinement of the proof of  Theorem~\ref{thm:ges-as-ese} above. \end{proof}
 
 \begin{cor}
 The pseudo adjunction from  $\ESE$ to $\GES$ cuts down to a pseudo equivalence of categories between the subcategory of \edc's $\ESE^1$ satisfying axioms (A), (B), (C), (D$_1$) and the subcategory of $\GES$ comprising the replete general event structures.
  \end{cor}
 
 \end{document}